\documentclass[aps,amsmath,amssymb,twocolumn,nofootinbib,pra,superscriptaddress]{revtex4-2}

\usepackage{graphicx,tikz}
\usepackage{dcolumn}
\usepackage{bm}
\usepackage{epsfig,hyperref,amsthm}
\usepackage{mathtools}
\usepackage{color,fancybox}
\usepackage{colordvi}
\usepackage{relsize}
\usepackage{accents}
\usepackage{wasysym} 
\usepackage{xypic}
\usepackage{enumitem}
\usepackage{mathtools,slashed}
\usepackage{times}
\usepackage{newtxmath}

\hypersetup{
	colorlinks=true,
	linkcolor=CitingColor,
	citecolor=CitingColor,
	urlcolor=CitingColor
}

\DeclareMathAlphabet\mathbfcal{OMS}{cmsy}{b}{n}
\newcommand{\ket}[1]{\ensuremath{|#1\rangle}}
\newcommand{\bra}[1]{\ensuremath{\langle #1|}}
\newcommand{\braket}[2]{\langle #1|#2\rangle}
\newcommand{\proj}[1]{\ket{#1}\bra{#1}}
\newcommand{\be}{\begin{equation}}
\newcommand{\ee}{\end{equation}}
\newcommand{\ba}{\begin{eqnarray}}
\newcommand{\ea}{\end{eqnarray}}

\newcommand{\norm}[1]{\left\|#1\right\|}

\newcommand{\id}{\mathbb{I}}

\newtheorem{alemma}{Lemma}
\newtheorem{aproposition}{Proposition}
\newtheorem{atheorem}{Theorem}
\newtheorem{afact}{Fact}
\newtheorem{acorollary}{Corollary}
\newtheorem{adefinition}{Definition}

\newtheorem{question}{Question}

\definecolor{nred}{rgb}{0.9,0.1,0.1}
\definecolor{nblack}{rgb}{0,0,0}
\definecolor{nblue}{rgb}{0.2,0.2,0.8}
\definecolor{ngreen}{rgb}{0.2,0.6,0.2}
\definecolor{ublue}{rgb}{0,0,0.5}

\definecolor{pur}{rgb}{0.75,0,0.5}
\definecolor{nngrn}{rgb}{0,0.5,0.5}
\definecolor{CitingColor}{rgb}{0,0.3,1}

\newcommand{\blu}{\color{nblue}}

\newcommand{\CY}[1]{{\color{black}#1}}
\newcommand{\CYtwo}[1]{{\color{black}#1}}
\newcommand{\CYnew}[1]{{\color{black}#1}}
\newcommand{\CYthree}[1]{{\color{black}#1}}

\begin{document}
\title{Dynamical Landauer principle: Thermodynamic criteria of transmitting classical information}

\author{Chung-Yun Hsieh}
\email{chung-yun.hsieh@bristol.ac.uk}
\affiliation{H.H. Wills Physics Laboratory, University of Bristol, Tyndall Avenue, Bristol BS8 1TL, United Kingdom}
\affiliation{ICFO - Institut de Ci\`encies Fot\`oniques, The Barcelona Institute of Science and Technology, 08860 Castelldefels, Spain}

\date{\today}

\begin{abstract}
Transmitting energy and information are two essential aspects of nature.
Recent findings suggest they are closely related, while a quantitative equivalence between them is still unknown.
This thus motivates us to ask: {\em Can information transmission tasks equal certain energy transmission tasks?}
We answer this question positively by bounding various one-shot classical capacities via different energy transmission tasks.
Such bounds provide the physical implication that, in the one-shot regime, transmitting $n$ bits of classical information {\em is equivalent to} $n\times k_BT\ln2$ transmitted energy.
Unexpectedly, these bounds further uncover a dynamical version of Landauer's principle, showing the strong link between {\em transmitting} (rather than {\em erasing}) information and energy.
Finally, in the asymptotic regime, our findings further provide thermodynamic meanings for Holevo-Schumacher-Westmoreland Theorem and a series of strong converse properties as well as no-go theorems.
\end{abstract}

\maketitle

\section{Introduction}
Transmitting energy and information are two essential aspects of our everyday lives.
They are not just foundations of nature's functionalities but also key underpinnings of the broad sciences and technologies.
Even though they seem to be unrelated, several hints have suggested the {\em opposite}.
For instance, photons (i.e., light's quantised energy) can send classical messages, meaning that transmitting energy can provide information transmission.
On the other hand, the thermodynamic effects in transmitting/maintaining information~\cite{Plenio1999,Plenio2001,Schumacher2002,Maruyama2005,Hsieh2020,Hsieh2021,Narasimhachar2019,Korzekwa2019,Biswas2021,Auffeves2022} and the energy cost of information processing~\cite{Faist2015,Faist2018,Chiribella2021,Chiribella2022} jointly suggest that transmitting information may potentially be accompanied by certain types of energy transmission.
Still, a clear, quantitative equivalence between transmitting information and energy is still missing in the literature.
This thus motivates us to ask the following question:
\begin{center}
{\em {\bf(Central Question)} Can information transmission tasks be equivalent to certain energy transmission tasks?}
\end{center}
A suitable answer to the above question can uncover the foundational link between transmitting information and energy.

This work answers this question by proving the first such equivalence.
We first formulate information transmission via various one-shot classical communication tasks.
In such tasks, the ability to send classical information is quantified by different types of one-shot classical capacities.
Utilising entropic quantities introduced in Refs.~\cite{Aberg2013,Wang2013}, we prove entropic bounds on these one-shot classical capacities (Theorems~\ref{AppThm:MainResult} and~\ref{AppThm:CC}).
Then, we introduce a novel class of (one-shot) energy transmission tasks, whose figure-of-merits are equivalent to the one-shot classical capacities (Theorem~\ref{Result:TCTCI})---this thus answers the central question.
Surprisingly, Theorem~\ref{Result:TCTCI} provides an unexpected application---a dynamical version of {\em Landauer's principle}~\cite{Landauer1961} (Corollaries~\ref{coro:weak dynamical Landauer} and~\ref{coro:strong dynamical Landauer}), which largely strengthens the finding reported in Ref.~\cite{Hsieh2021}.
Finally, we show that Theorem~\ref{Result:TCTCI} can reproduce the {\em Holevo-Schumacher-Westmoreland} (HSW) Theorem~\cite{Holevo1973,Holevo1998,Schumacher1997} in the asymptotic regime, further revealing its thermodynamics meaning (Proposition~\ref{AppResult:HSWThermo}) and several no-go results (Corollaries~\ref{coro:strong converse property chi} and~\ref{coro:no-go}).
Please also see Fig.~\ref{Fig:summary}, which schematically explains how this paper is structured.

This work is the companion paper of Ref.~\cite{Companion2}.
As we aim to bridge the communities of thermodynamics and quantum communication, we detail the mathematical frameworks and provide thorough, step-by-step proofs of all results for a pedagogical purpose. 
The companion paper~\cite{Companion2} focuses more on our results' physical implications.

\section{Framework}\label{App:Proof-MainResult}

\begin{figure*}
\begin{center}
\scalebox{0.8}{\includegraphics{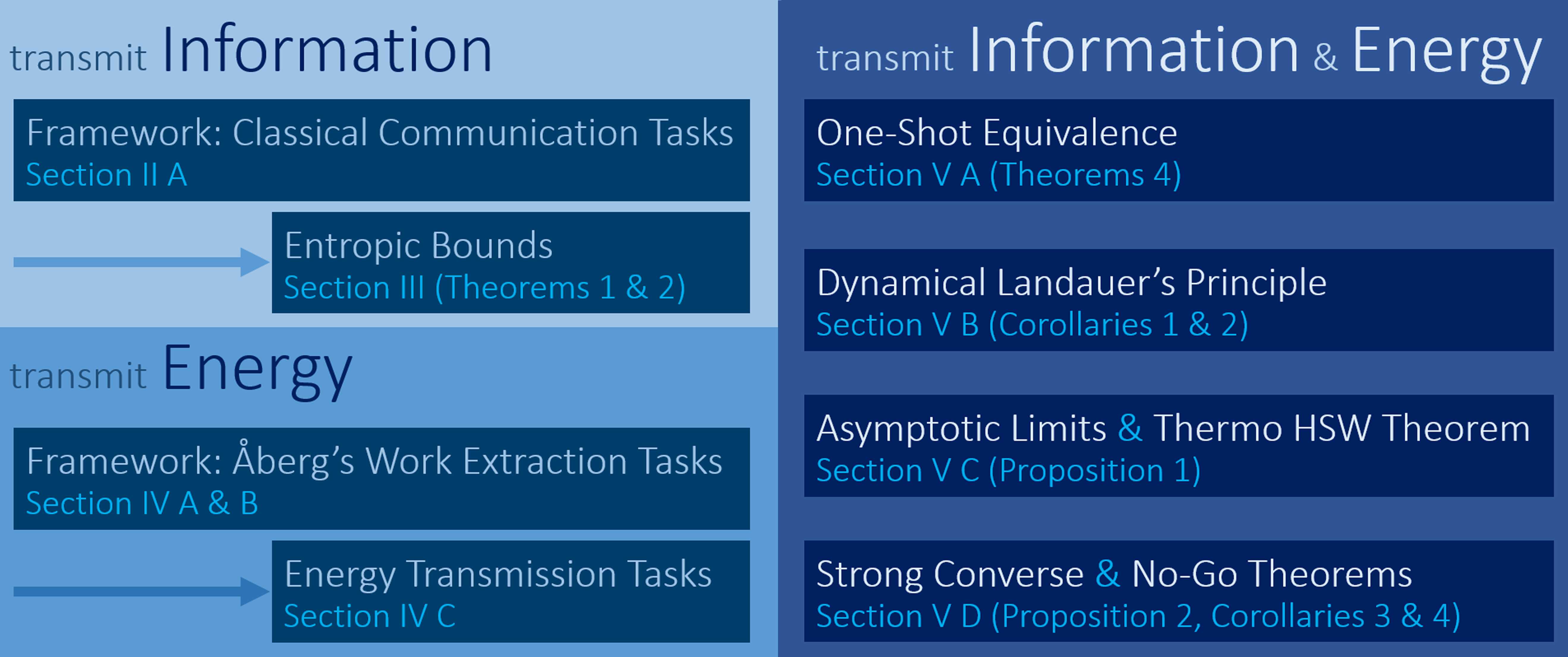}}
\caption{
{\bf Summary of this work.}
This paper is structured as follows.
Section~\ref{App:Proof-MainResult} contains preliminary notions, including the framework of classical communication tasks.
Section~\ref{Sec:Entropic bound on CC} provides entropic bounds on one-shot classical capacities (Theorems~\ref{AppThm:MainResult} and~\ref{AppThm:CC}).
Section~\ref{Sec:energy transmission} details the energy transmission tasks.
Section~\ref{Sec:TCTCI} contains results bridging information and energy transmission. 
Section~\ref{App:Proof} provides thermodynamic bounds on one-shot classical capacities (Theorem~\ref{Result:TCTCI}). 
Section~\ref{Sec:Landauer} uncovers the dynamical version of Landauer's principle (Corollaries~\ref{coro:weak dynamical Landauer} and~\ref{coro:strong dynamical Landauer}).
Section~\ref{App:AsymptoticLimit} provides the thermodynamic meaning of the HSW Theorem (Proposition~\ref{AppResult:HSWThermo}).
Section~\ref{Sec:no-go} reports strong converse properties and no-go results (Corollaries~\ref{coro:strong converse property chi} and~\ref{coro:no-go}).
Section~\ref{Sec:Conclusion} concludes the paper.
}
\label{Fig:summary} 
\end{center}
\end{figure*}

In this paper, we always consider quantum systems with finite dimensions.
We now start with a quick recap of basic notions from quantum information theory.

First, for a given quantum system (denoted by $S$), a {\em quantum state}, or simply {\em state} (also known as {\em mixed state} or {\em density matrix}), is a semi-definite positive operator $\rho\ge0$ acting on $S$ with ${\rm tr}(\rho)=1$~\cite{QIC-book}.
By the spectral decomposition theorem~\cite{QIC-book}, every (mixed) state can be written as a convex mixture of {\em pure states}~\footnote{A state $\rho$ is {\em pure} if ${\rm tr}(\rho^2)=1$~\cite{QIC-book}. In this case, we will write $\rho = \proj{\phi}$, where $\ket{\phi}$ is a so-called ket vector.} as $\rho = \sum_ip_i\proj{\phi_i}$ with \mbox{$p_i\ge0$}, $\sum_ip_i=1$, and $\braket{\phi_i}{\phi_j}=\delta_{ij}$.
Physically, this means that such a (mixed) state can be prepared by generating the pure state $\ket{\phi_i}$ with probability $p_i$ in a multi-trial experiment---it is a ``statistical'' mixture of pure states.

Second, a physical measurement can be described by a {\em positive operator-valued measure} (POVM)~\cite{QIC-book}, which is a set $\{E_m\}_m$ of operators acting on $S$ with $E_m\ge0$ and $\sum_mE_m=\id_S$ ($\id_S$ is the identity operator acting on $S$).
For an input state $\rho$, it describes the process that the measurement outputs the $m$-th outcome with probability ${\rm tr}(E_m\rho)$.

Third, quantum dynamics can be described by {\em channels}, which are {\em completely-positive trace-preserving linear maps}~\cite{QIC-book}.
For an input state $\rho$, a channel $\mathcal{N}$ describes the process $\rho\mapsto\mathcal{N}(\rho)$, where $\mathcal{N}(\rho)$ is the output state.
The notion of channels provides a mathematical way to describe general quantum information processing.

Finally, deterministic ways to manipulate a channel are described by {\em superchannels}~\cite{Chiribella2008,Chiribella2008-2}, which are linear maps bringing a channel to another channel. For an initial channel $\mathcal{N}$, a superchannel $\Pi$ outputs another channel denoted as $\Pi(\mathcal{N})$.
Mathematically, as characterised by Refs.~\cite{Chiribella2008,Chiribella2008-2}, every physically relevant superchannel $\Pi$ can be written as
$
\Pi(\mathcal{N})= \mathcal{E}_{\rm post}\circ(\mathcal{N}\otimes\mathcal{I}_{\rm auxiliary})\circ\mathcal{E}_{\rm pre},
$
where $\mathcal{E}_{\rm pre}$ and $\mathcal{E}_{\rm post}$ are some channels, and $\mathcal{I}_{\rm auxillary}$ is the identity channel acting on a finite-dimensional auxiliary system.
Physically, this means that all (physically relevant) superchannels can be realised by adding some pre- and post-processing channels ($\mathcal{E}_{\rm pre}$ and $\mathcal{E}_{\rm post}$) to the given channel ($\mathcal{N}$) with the help of some finite-size auxiliary systems.
In this work, we will use the notation ``$\Theta$'' to denote a generic set of superchannels.
A simple yet important example is the one that only contains the identity superchannel, which can be written as
\begin{align}\label{Eq:standard C}
\Theta =\Theta_{\rm C} \coloneqq \{(\cdot)\mapsto(\cdot)\}.
\end{align}
This special set of superchannels will play a crucial role later.

\begin{figure*}
\begin{center}
\scalebox{0.8}{\includegraphics{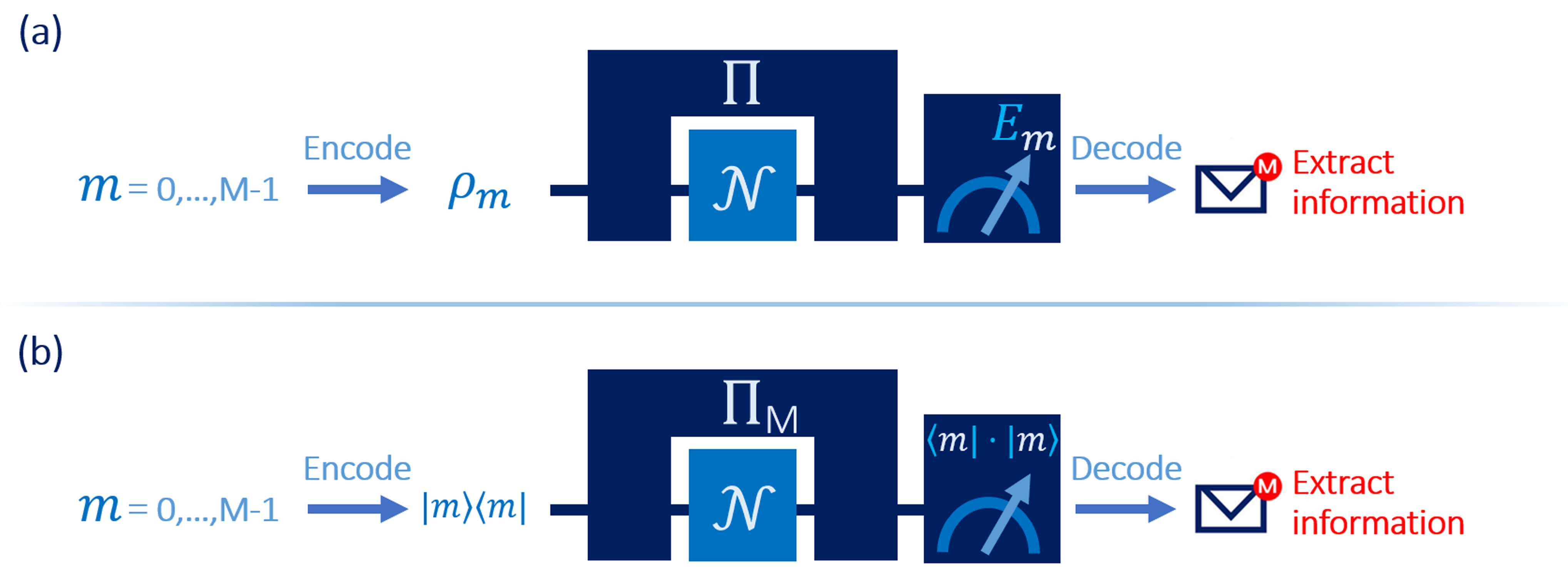}}
\caption{
\CY{\bf Two equivalent formulations of classical communication.}
\CY{(a) The task corresponding to Definition~\ref{AppDef:CCapacity}. The sender encodes the classical index $m$ into the state $\rho_m$. After sending it via $\Pi(\mathcal{N})$ for some $\Pi\in\Theta$, the receiver decodes it by the measurement $\{E_{m'}\}_{m'=0}^{M-1}$. The communication is successful if $m'=m$.
(b) The task corresponding to Fact~\ref{AppLemma:AlternativeCCapacity} (which is equivalent to Definition~\ref{AppDef:CCapacity}). The sender encodes $m$ into a pre-fixed computational basis $\ket{m}$.
Then, they send it via a classical version of $\mathcal{N}$ assisted by $\Theta$ as given in Definition~\ref{Def:classical_ver}; namely, the classical-to-classical channel $\Pi_M(\mathcal{N})$. Finally, the receiver decodes the information by applying the projective measurement $\{\proj{m'}\}_{m'=0}^{M-1}$.}
}
\label{Fig:CCtasks} 
\end{center}
\end{figure*}

\subsection{Classical Communication via Quantum Channels}\label{App:Proof-MainResult-Preliminary}
With the notions just introduced, we can now detail the {\em $\Theta$-assisted classical communication task}, containing a sender and a receiver [see also Fig.~\ref{Fig:CCtasks} (a)].
The sender's goal is to send a set of classical indices $\{m\}_{m=0}^{M-1}$ (termed {\em classical data} or {\em classical information}) to the receiver via a given channel $\mathcal{N}$.
Whenever such a set can be reliably sent to the receiver in a multi-trial experiment, the sender can ``inform'' the receiver of certain nontrivial things (e.g. the sender's date of birth, age, or anything that can be represented by a finite combination of integers $m$'s).
In other words, the sender can {\em communicate} with the receiver.
To do so via the quantum channel $\mathcal{N}$, the sender will first choose a set of quantum states $\{\rho_m\}_{m=0}^{M-1}$ to represent the classical indices $m$'s---this is the so-called {\em encoding}.
That is, the sender {\em encodes} the classical index $m$ into some quantum state $\rho_m$, making it a valid input for the channel.
Without any additional assistance, the sender will send $\rho_m$ to the receiver through $\mathcal{N}$. Then, to extract the classical indices from $\mathcal{N}$'s output (which is still quantum), the receiver measures the output $\mathcal{N}(\rho_m)$ via some POVM $\{E_{m'}\}_{m'=0}^{M-1}$.
The communication is successful if the measurement outcome $m'$ coincides with the originally encoded index $m$; namely, when $m'=m$.
This is the so-called {\em decoding}.

Now, suppose the sender and receiver are allowed to use some superchannels to assist their communication, and the allowed superchannels are collectively described by the set $\Theta$.
Mathematically, this means that the sender can send $\rho_m$'s to the receiver through a channel $\Pi(\mathcal{N})$ with some $\Pi\in\Theta$.
Hence, a successful communication corresponds to the following transformation:
\begin{align}
m\mapsto{\rm tr}[E_m\Pi(\mathcal{N})(\rho_m)]\quad\forall\,m.
\end{align}
We call this a {\em $\Theta$-assisted scenario}.
In the literature, a commonly used measure of a channel's performance in a communication task is the so-called capacity.
Now, we can define this measure for a $\Theta$-assisted scenario, which quantifies $\mathcal{N}$'s ability to transmit (classical) information with the assistance of superchannels from $\Theta$ (since now, $\mathbb{N}$ denotes the set of positive integers):

\begin{adefinition}\label{AppDef:CCapacity}
{\em
{(One-Shot $\Theta$-Assisted Classical Capacity)}
For a given set of superchannels $\Theta$ and an error parameter \mbox{$0\le\epsilon\le1$}, the {\em one-shot $\Theta$-assisted classical capacity subject to an error $\epsilon$} of a given channel $\mathcal{N}$ is defined by
\begin{eqnarray}
\begin{aligned}
C_{\Theta,(1)}^\epsilon(\mathcal{N})\coloneqq\max_{\substack{M,\Pi,\\\{\rho_m\},\{E_m\}}}&\log_2M\\
{\rm s.t.}\quad&M\in\mathbb{N};\;\Pi\in\Theta;\\
&\{\rho_m\}_{m=0}^{M-1}\;\text{is a set of states};\\
&\{E_m\}_{m=0}^{M-1}\;\text{is a POVM};\\
&\sum_{m=0}^{M-1}\frac{1}{M}{\rm tr}\left[E_m\Pi(\mathcal{N})(\rho_m)\right]\ge1-\epsilon.
\end{aligned}
\end{eqnarray}
}
\end{adefinition}
As mentioned in Ref.~\cite{Wang2013}, one-shot classical capacity provides richer knowledge than the {\em asymptotic} classical capacity. This is because applying the multi-copy limit can reproduce the asymptotic result (more details in Sec.~\ref{App:AsymptoticLimit}).
This explains our motivation to use $C_{\Theta,(1)}$ as the figure-of-merit to measure a channel's ability to transmit information.

The classical communication tasks mentioned above have an alternative formulation, which is crucial for this work.
To introduce it, we need to make sense of a channel's ``classical version.''
First, from now on, we will use the symbol $M$ (and, in some limited cases, $L$) to denote the size of classical messages. 
It can be given by the dimension of a system spanned by a pre-defined orthonormal basis $\{\ket{m}\}_{m=0}^{M-1}$.
Any state diagonal in this given basis, which is of the form $\sum_{m=0}^{M-1}p_m\proj{m}$, can be equivalently described by classical probability distributions $\{p_m\}_{m=0}^{M-1}$.
Due to this reason, we call such a system a {\em classical system}.
For a classical system with dimension $M$, we again use the symbol ``$M$'' to denote the system.
Note that when we say a system $M$ is ``classical,'' it is understood that an orthonormal basis $\{\ket{m}\}_{m=0}^{M-1}$ has been assigned, and we mainly focus on states in this system diagonal in the given basis $\{\ket{m}\}_{m=0}^{M-1}$ (even though, ultimately, it is still physically a ``quantum'' system).
This notion helps us to define {\em classical behaviours} of quantum channels as follows:

\begin{adefinition}{\em
{(Classical-to-Quantum and Quantum-to-Classical Channels)}\label{Def:classical_ver}
${\rm CQ}_{M\to S}$ denotes the set of all {\em classical-to-quantum channels} from a classical system $M$ to a system $S$ of the form 
\begin{align}\label{AppEq:CQ}
{(\cdot)_M\mapsto}\sum_{m=0}^{M-1}\rho_{m|S}\bra{m}(\cdot)_M\ket{m}_M,
\end{align} 
where $\{\rho_{m|S}\}_{m=0}^{M-1}$ are states in $S$.

${\rm QC}_{S \to M}$ denotes the set of all {\em quantum-to-classical channels} from a system $S$ to a classical system $M$ of the form 
\begin{align}\label{AppEq:QC}
{(\cdot)_S\mapsto}\sum_{m=0}^{M-1}\proj{m}_M{\rm tr}[E_{m|S}(\cdot)_S],
\end{align} 
where $\{E_{m|S}\}_{m=0}^{M-1}$ is a POVM in $S$.
}
\end{adefinition}
Finally, a channel is called a {\em classical-to-classical channel} in a classical system $M$, or simply a {\em $M$-to-$M$ classical channel}, if it can be written as $\mathcal{L}\circ\mathcal{K}$ for some $\mathcal{K}\in{\rm CQ}_{M\to S}$ and $\mathcal{L}\in{\rm QC}_{S \to M}$.
Using Definition~\ref{Def:classical_ver}, we can now define the {\em classical versions} of the set $\Theta$:

\begin{adefinition}\label{AppDef:ClassicalVersion}
{\em
{(Classical Versions of Superchannels)}
Given a set of superchannels $\Theta$ and a classical system $M$.
The set of {\em $M$-to-$M$ classical versions of $\Theta$}, denoted by $\Theta_M$, is given by
\begin{align}
&\Theta_M\coloneqq\nonumber\\
&\left\{\mathcal{L}\circ[\Pi(\cdot)]\circ\mathcal{K}\,\middle|\,\Pi\in\Theta,\mathcal{K}\in{\rm CQ}_{M\to S_{\rm in|\Pi}},\mathcal{L}\in{\rm QC}_{S_{\rm out|\Pi} \to M}\right\}.
\end{align}
$S_{\rm in|\Pi}$ ($S_{\rm out|\Pi}$) is the input (output) space of $\Pi$'s output channel.
}
\end{adefinition}
The set $\Theta_M$ characterises all possible classical-to-classical realisations induced by $\Theta$ in a $M$-dimensional setting.
Hence, every element $\Pi_M\in\Theta_M$ maps \CY{a channel} to some $M$-to-$M$ classical channel.
\CY{For this reason, we also call the mapping $\Pi_M(\mathcal{N})$ a ($M$-to-$M$) {\em classical version of $\mathcal{N}$} ({\em assisted by $\Theta$})~\footnote{\CY{We remark that the channels \mbox{``$\mathcal{E}_{{\rm en}|M}:A\to S_{\rm in}$''} and ``$\mathcal{E}_{{\rm de}|M}:S_{\rm out}\to A$'' defined in the companion paper~\cite{Companion2} are in ${\rm CQ}_{M\to S_{\rm in}}$ and ${\rm QC}_{S_{\rm out} \to M}$.
}}.}

\CY{
Now, we can rewrite Definition~\ref{AppDef:CCapacity} as follows:}

\begin{afact}\label{AppLemma:AlternativeCCapacity}
{\em(Alternative Form of One-Shot $\Theta$-Assisted Classical Capacity)}
\CY{With} the same setting as in Definition~\ref{AppDef:CCapacity}, \CY{we have}
\CY{\begin{eqnarray}
\begin{aligned}\label{AppEq:CCapacityDef}
C_{\Theta,(1)}^\epsilon(\mathcal{N})\coloneqq\max_{M,\Pi_M}&\log_2M\\
{\rm s.t.}\quad&M\in\mathbb{N};\;\Pi_M\in\Theta_M;\\
&P_s[\Pi_M(\mathcal{N})]\ge1-\epsilon,
\end{aligned}
\end{eqnarray}}
where \CY{the average success probability reads}
\begin{align}
\CY{P_s[\Pi_M(\mathcal{N})]\coloneqq\sum_{m=0}^{M-1}\frac{1}{M}\bra{m}\Pi_M(\mathcal{N})(\proj{m})\ket{m}}.
\end{align}
\end{afact}
\begin{proof}
The validity of Eq.~\eqref{AppEq:CCapacityDef} can be seen by noting that
\begin{itemize}
\item For every set of states $\{\rho_{m|S}\}_{m=0}^{M-1}$ in $S$, there exists a channel $\mathcal{K}_{M\to S}\in{\rm CQ}_{M\to S}$ given by Eq.~\eqref{AppEq:CQ} such that 
\mbox{$
\mathcal{E}(\rho_m) = \mathcal{E}\circ\mathcal{K}_{M\to S}(\proj{m})
$}
for every index $m$ and channel $\mathcal{E}$.
\item For every POVM $\{E_{m|S}\}_{m=0}^{M-1}$ in $S$, there exists a channel $\mathcal{L}_{S\to M}\in{\rm QC}_{S\to M}$ given by Eq.~\eqref{AppEq:QC} such that 
\mbox{$
{\rm tr}\left(E_{m|S}\rho\right) = \bra{m}\mathcal{L}_{S\to M}(\rho)\ket{m}
$}
for every index $m$ and state $\rho$.
\end{itemize}
\CY{Hence,} from Definition~\ref{AppDef:CCapacity}, \CY{we have} ($S_{\rm in|\Pi}$ \CY{and} $S_{\rm out|\Pi}$ \CY{are} the input and output spaces of \CY{$\Pi$'s} output channel, \CY{respectively})
\CY{
\begin{eqnarray}
\begin{aligned}\label{Eq:useful formula C Theta}
C_{\Theta,(1)}^\epsilon(\mathcal{N})\coloneqq\max_{M,\Pi,\mathcal{K},\mathcal{L}}&\log_2M\\
{\rm s.t.}\quad&M\in\mathbb{N};\;\Pi\in\Theta;\\
&\mathcal{K}\in{\rm CQ}_{M\to S_{\rm in|\Pi}};\;\mathcal{L}\in{\rm QC}_{S_{\rm out|\Pi}\to M};\\
&P_s[\mathcal{L}\circ\Pi(\mathcal{N})\circ\mathcal{K}]\ge1-\epsilon,
\end{aligned}
\end{eqnarray}}
and the claim follows by using Definition~\ref{AppDef:ClassicalVersion}.
\end{proof}
\CY{An important example is when we allow {\em no additional assistance}; namely, $\Theta = \Theta_{\rm C}$ as given in Eq.~\eqref{Eq:standard C}.
Then,} we have $C_{\Theta_{\rm C},(1)}^\epsilon = C_{(1)}^\epsilon$, which is the (standard) one-shot classical capacity~\cite{Wang2013}.
\CY{It measures the channel's ``primal''} ability to transmit classical information in the one-shot regime.
\CY{Hence, intuitively,} every $\Theta$-assisted scenario should be linked to this fundamental case.
\CY{Indeed, using Fact~\ref{AppLemma:AlternativeCCapacity} and Eq.~\eqref{Eq:useful formula C Theta}, we have}
\begin{widetext}
\begin{align}
\sup_{\Pi\in\Theta}C_{(1)}^\epsilon[\Pi(\mathcal{N})]&\coloneqq\sup_{\Pi\in\Theta}\max\left\{\log_2M\,\middle|\,\exists\,\mathcal{K}\in{\rm CQ}_{M\to S_{\rm in|\Pi}},\mathcal{L}\in{\rm QC}_{S_{\rm out|\Pi}\to M}\;{\rm s.t.}\;P_s[\mathcal{L}\circ\Pi(\mathcal{N})\circ\mathcal{K}]\ge1-\epsilon\right\}\nonumber\\
&=\max\{\log_2M\,|\,\exists\,\Pi_M\in\Theta_M\;{\rm s.t.}\;P_s[\Pi_M(\mathcal{N})]\ge1-\epsilon\}=C_{\Theta,(1)}^\epsilon(\mathcal{N}).
\end{align}
\end{widetext}
\CY{We thus obtain the following useful fact:}
\begin{afact}\label{AppLemma:C-CTheta}
{\em(Standard Classical Capacity Representation of $\Theta$-Assisted Capacity)}
For a given set of superchannels $\Theta$, a channel $\mathcal{N}$, and an error $0\le\epsilon\le1$, we have that
\begin{align}
\CY{C_{\Theta,(1)}^\epsilon(\mathcal{N}) = \sup_{\Pi\in\Theta}C_{(1)}^\epsilon[\Pi(\mathcal{N})].}
\end{align}
\end{afact}
\CY{We thus have a thorough mathematical framework to analyse and quantify a channel's ability to transmit information.}

\CY{\subsection{Thermodynamics and Informational Non-equilibrium}\label{Sec:Thermo QRT}}

\CY{
As we aim to bridge communication and thermodynamics, we briefly review key ingredients from the resource-theoretic approach to thermodynamics (see, e.g., Ref.~\cite{Lostaglio2019} for a pedagogical review).
To start with, consider a given $d$-dimensional system (with $d<\infty$).
A {\em quantum resource theory}, or simply {\em resource theory}, is a pair of sets $(\mathcal{F}_R,\mathcal{O}_R)$, where ``$R$'' denotes the given resource. $\mathcal{F}_R$ is called the set of {\em free states}, which contains all states that do not possess the given resource $R$.
$\mathcal{O}_R$ is the set of {\em allowed operations} (also known as {\em free operations} and {\em available operations}), which are physically allowed ways to manipulate the given resource $R$.
Crucially, when an allowed operation is a channel, a well-accepted necessary condition is that it {\em cannot generate $R$ from any free state}; namely,
\begin{align}\label{Eq:golden rule}
\mathcal{E}(\eta)\in\mathcal{F}_R\quad\forall\,\eta\in\mathcal{F}_R\;\&\;\mathcal{E}\in\mathcal{O}_R.
\end{align}
That is, such channels cannot generate useful resources from nothing.
This is sometimes called the {\em golden rule} of quantum resource theories~\cite{ChitambarRMP2019}.
For instance, when we set $R$ as entanglement, then $\mathcal{F}_R$ is the set of all separable states, and a valid option of $\mathcal{O}_R$ is the set of all local operations and classical communication channels, which satisfies Eq.~\eqref{Eq:golden rule}.
We refer the reader to Ref.~\cite{ChitambarRMP2019} for a general review.
Here, we focus on its application to thermodynamics---when we set $R$ as the {\em status of non-equilibrium}, also known as {\em athermality}.}

\CY{
To mathematically describe athermality, we first need to know how to describe the thermal equilibrium state, termed thermal state.
Formally, with a given background temperature $T$ and system Hamiltonian $H$, the {\em thermal state} is defined by
\begin{align}\label{Eq: thermal state}
\gamma_H \coloneqq \frac{e^{-H/k_BT}}{{\rm tr}\left(e^{-H/k_BT}\right)}, 
\end{align}
where $k_B$ is the Boltzmann constant. 
It describes a system in thermal equilibrium according to the Boltzmann distribution.
In this work, we always consider a fixed temperature; hence, we only keep the Hamiltonian dependence explicit.
Now, since every state that is not thermal contains certain non-equilibrium effect, the resource theory of thermodynamics has a unique free state $\gamma_H$, which is the {\em only} member in the set $\mathcal{F}_R$.
Hence, Eq.~\eqref{Eq:golden rule} now takes a simple form, which defines the so-called {\em Gibbs-preserving channels}---channels keeping thermal equilibrium untouched, a common type of allowed operations for thermodynamics:
\begin{align}
\mathcal{E}(\gamma_H) = \gamma_H.
\end{align}
Notably, this is only a thermodynamic constraint for processing information, while manipulating Hamiltonians is also crucial in thermodynamics.
This type of physical manipulations will be addressed in Sec.~\ref{App:Proof-MainResult-WorkExtraction}. 
For now, we focus on the thermodynamic effects of information processing.

A special case of athermality is when $H=0$. 
Namely, we turn off the energy differences in $H$ to isolate the informational contribution to thermodynamics.
In this case, we have $\gamma_{H=0} = \id/d$, and any state that is not thermal is in the so-called {\em informational non-equilibrium}~\cite{Gour2015,Stratton2023,Hsieh2024-3}.
This notion of non-equilibrium is of particular importance for this work since it characterises how information influences thermodynamics.
One example is the following observation:}\\

\begin{afact}\label{Fact:A1}
{\em(Preserving the Gibbs-Preserving Property)}
\CY{Consider} a set of superchannels $\Theta$, a channel $\mathcal{N}$, and an error \mbox{$0\le\epsilon\le1$}.
Then, for every $\Pi_M\in\Theta_M$ satisfying \mbox{$P_s[\Pi_M(\mathcal{N})]\ge1-\epsilon$}, we have that 
\begin{align}
\CY{\norm{\Pi_M(\mathcal{N})\left(\frac{\id_M}{M}\right) - \frac{\id_M}{M}}_1\le2\epsilon,}
\end{align} 
\CY{where, for an operator $A$, its trace norm is defined as~\cite{QIC-book}}
\begin{align}\label{Eq:trace norm}
\CY{\norm{A}_1\coloneqq{\rm tr}\left(\sqrt{A^\dagger A}\right).}
\end{align}
\end{afact}
In other words, when the system Hamiltonian is fully degenerate \CY{and $\Pi_M$ is ``good enough'' in communication in the sense that \mbox{$P_s[\Pi_M(\mathcal{N})]\ge1-\epsilon$}, then} $\Pi_M$ maps the channel $\mathcal{N}$ to some output channel that is \CY{``almost''} Gibbs-preserving.
\begin{proof}
For such a $\Pi_M\in\Theta_M$, direct computation shows that
\begin{widetext}
\begin{align}
\norm{\Pi_M(\mathcal{N})\left(\frac{\id_{M}}{M}\right) - \frac{\id_{M}}{M}}_1&=\norm{\frac{1}{M}\sum_{n=0}^{M-1}\left[\sum_{m=0}^{M-1}\bra{n}\Pi_M(\mathcal{N})(\proj{m})\ket{n} - 1\right]\proj{n}}_1\nonumber\\
&=\norm{\frac{1}{M}\sum_{n=0}^{M-1}\left[\left(\bra{n}\Pi_M(\mathcal{N})(\proj{n})\ket{n} - 1\right) + \sum_{\substack{m=0\\m\neq n}}^{M-1}\bra{n}\Pi_M(\mathcal{N})(\proj{m})\ket{n}\right]\proj{n}}_1\nonumber\\
&=\frac{1}{M}\sum_{n=0}^{M-1}\left|\left(\bra{n}\Pi_M(\mathcal{N})(\proj{n})\ket{n} - 1\right) + \sum_{\substack{m=0\\m\neq n}}^{M-1}\bra{n}\Pi_M(\mathcal{N})(\proj{m})\ket{n}\right|\nonumber\\
&\le\frac{1}{M}\sum_{n=0}^{M-1}\left[\left|\bra{n}\Pi_M(\mathcal{N})(\proj{n})\ket{n} - 1\right| + \sum_{\substack{m=0\\m\neq n}}^{M-1}\bra{n}\Pi_M(\mathcal{N})(\proj{m})\ket{n}\right]\nonumber\\
&=2\left(1 - \frac{1}{M}\sum_{n=0}^{M-1}\bra{n}\Pi_M(\mathcal{N})(\proj{n})\ket{n}\right) = 2\left(1-P_s[\Pi_M(\mathcal{N})]\right)\le2\epsilon.
\end{align}
\end{widetext}
In the first line, we have used the fact that the output of $\Pi_M(\mathcal{N})$ is diagonal in the basis $\{\ket{m}\}_{m=0}^{M-1}$.
\end{proof}

\CY{With the framework of communication formalised in this section, we can formally prove this work's first main result, which are entropic bounds on the one-shot $\Theta$-assisted classical capacities, as detailed in the following section.}

\CY{
\section{Entropic Bounds on Classical Capacities
}\label{Sec:Entropic bound on CC}
}

\CY{
\subsection{Bounding One-Shot $\Theta$-Assisted Classical Capacities}
}
\CY{This section aims to} show upper and lower bounds on \CY{the one-shot $\Theta$-assisted} classical \CY{capacities via two closely related entropic quantities, as defined below.} 
The first one is the {\em $\epsilon$-smoothed relative R\'enyi $0$-entropy} (see Supplementary Definition 6 in Ref.~\cite{Aberg2013}) defined for two commuting states $\eta=\sum_jq_j\proj{j},\CY{\xi}=\sum_jr_j\proj{j}$ as 
\begin{align}\label{AppEq:D0epsilon}
D_0^\CY{\epsilon}(\eta\,\|\,\CY{\xi})\coloneqq\max_{\Lambda:\sum_{j\in\Lambda}q_j>1-\CY{\epsilon}}\log_2\frac{1}{\sum_{j\in\Lambda}r_j},
\end{align}
where the maximisation is taken over every possible index set $\Lambda$ satisfying the strict inequality $\sum_{j\in\Lambda}q_j>1-\CY{\epsilon}$.
\CY{In Appendix~\ref{App}, we discuss and prove several mathematical properties of $D_0^\epsilon$, including its relation with the so-called min-relative entropy~\cite{Datta2013}, data-processing inequality, smoothness and continuity.}
This entropy can be extended to the {\em hypothesis testing relative entropy} with error $\epsilon$, which is defined as follows for two states $\rho,\sigma$~\cite{Wang2013}:
\begin{align}\label{AppEq:HypothesisTesting}
D_h^\epsilon(\rho\,\|\,\sigma)\coloneqq\max_{\substack{0\le Q\le\id\\{\rm tr}(Q\rho)\ge1-\epsilon}}\log_2\frac{1}{{\rm tr}(Q\sigma)}.
\end{align}
As a direct observation \CY{from} Eq.~\eqref{AppEq:D0epsilon}, one can see that 
\begin{align}\label{Eq:useful D_0 D_h}
\CY{D_0^\epsilon(\rho\,\|\,\sigma)\le D_h^\epsilon(\rho\,\|\,\sigma)}
\end{align} 
whenever they are both well-defined.
%
%
Now, we \CY{present this section's} main result.
In what follows, \CY{the symbol $M'$ denotes an auxiliary} classical system \CY{with the same dimension as $M$.}

\begin{widetext}

\begin{atheorem}\label{AppThm:MainResult} {\em(Bounding One-Shot $\Theta$-Assisted Classical Capacity)}
\CY{For} a set of superchannels $\Theta$, a channel $\mathcal{N}$, and errors $0<\delta\le\omega<\epsilon\le\CY{1/2}$, we have that 
\begin{align}\label{Eq:AppThm:MainResult}
&\sup_{\substack{M\in\mathbb{N},\Pi\in\Theta\\\eta_{S_{\rm in|\Pi}S'}=\sum_{x=0}^{M-1}p_x\sigma_{x|S_{\rm in|\Pi}}\otimes\kappa_{x|S'}}}D_h^{\omega}\left[(\Pi(\mathcal{N})\otimes\mathcal{I}_{S'})(\eta_{S_{\rm in|\Pi}S'})\,\middle\|\,\Pi(\mathcal{N})\left(\eta_{S_{\rm in|\Pi}}\right)\otimes\eta_{S'}\right]-\log_2\frac{4\epsilon}{(\epsilon - \omega)^2}\le\nonumber\\
&\quad\quad\quad\quad\quad\quad  C_{\Theta,(1)}^\epsilon(\mathcal{N})\le\sup_{\substack{M\in\mathbb{N},\Pi_M\in\Theta_M\\\norm{\Pi_M(\mathcal{N})\left(\frac{\id_M}{M}\right) - \frac{\id_M}{M}}_1\le2(\epsilon+\delta)}}D_0^{\epsilon+\delta}\left[(\Pi_M(\mathcal{N})\otimes\mathcal{I}_{M'})(\Phi_{MM'})\,\middle\|\,\Pi_M(\mathcal{N})\left(\frac{\id_{M}}{M}\right)\otimes\frac{\id_{M'}}{M}\right],
\end{align}
where $S'$ is some finite-dimensional \CY{auxiliary} system which is not necessarily of the same size with $S_{\rm in|\Pi}$, $\eta_{S_{\rm in|\Pi}S'}$ is a separable state in $S_{\rm in|\Pi}S'$ that can be written as a convex combination of $M$ product states (i.e., both $\sigma_{x|S_{\rm in|\Pi}}$, $\kappa_{x|S'}$ are states), and
\begin{align}\label{Eq: max corr state}
\CY{\Phi_{MM'}\coloneqq\frac{1}{M}\sum_{m=0}^{M-1}\proj{m}_{M}\otimes\proj{m}_{M'}}
\end{align} 
is the maximally (classically) correlated state diagonal in the bipartite basis associated with the classical systems $M,M'$.
\end{atheorem}
\CY{Note that} 
$[\Pi_M(\mathcal{N})\otimes\mathcal{I}_{M'}](\Phi_{MM'})$ and $\Pi_M(\mathcal{N})\left(\frac{\id_M}{M}\right)\otimes\frac{\id_{M'}}{M}$ are indeed simultaneously diagonalisable, \CY{and} the smoothed relative R\'enyi $0$-entropy is well-defined here.
Importantly, \CY{in Eq.~\eqref{Eq:AppThm:MainResult},} the upper bound in the second line is further {\em upper bounded by} the lower bound in the first line, up to the one-shot error term $-\log_2\CY{[4\epsilon/(\epsilon - \omega)^2]}$.
\CY{Hence,} both the upper and lower bounds will converge to the same quantity in the asymptotic limit \CY{(as detailed in Sec.~\ref{App:AsymptoticLimit})}.

As the first observation, suppose \CY{Theorem~\ref{AppThm:MainResult}} holds for the one-shot (standard) classical capacity $C_{(1)}$. 
Then Fact~\ref{AppLemma:C-CTheta} \CY{implies}
\begin{align}
C_{\Theta,(1)}^\epsilon(\mathcal{N})& = \sup_{\Pi\in\Theta}C_{(1)}^\epsilon[\Pi(\mathcal{N})]\nonumber\\
&\le\sup_{\Pi\in\Theta}\sup_{\substack{M\in\mathbb{N}\\\mathcal{K}\in{\rm CQ}_{M\to S_{\rm in|\Pi}},\\\mathcal{L}\in{\rm QC}_{S_{\rm out|\Pi}\to M},\\\norm{\mathcal{L}\circ\Pi(\mathcal{N})\circ\mathcal{K}\left(\frac{\id_M}{M}\right) - \frac{\id_M}{M}}_1\le2(\epsilon+\delta)}}D_0^{\epsilon+\delta}\left[((\mathcal{L}\circ\Pi(\mathcal{N})\circ\mathcal{K})\otimes\mathcal{I}_{M'})(\Phi_{MM'})\,\middle\|\,(\mathcal{L}\circ\Pi(\mathcal{N})\circ\mathcal{K})\left(\frac{\id_{M}}{M}\right)\otimes\frac{\id_{M'}}{M}\right]\nonumber\\
&\CY{=}\sup_{\substack{M\in\mathbb{N},\Pi_M\in\Theta_M\\\norm{\Pi_M(\mathcal{N})\left(\frac{\id_M}{M}\right) - \frac{\id_M}{M}}_1\le2(\epsilon+\delta)}}D_0^{\epsilon+\delta}\left[(\Pi_M(\mathcal{N})\otimes\mathcal{I}_{M'})(\Phi_{MM'})\,\middle\|\,\Pi_M(\mathcal{N})\left(\frac{\id_{M}}{M}\right)\otimes\frac{\id_{M'}}{M}\right].
\end{align}
On the other hand, \CY{utilising Fact~\ref{AppLemma:C-CTheta} again, Theorem~\ref{AppThm:MainResult}'s} lower bound can be rewritten as
\begin{align}
\sup_{\Pi\in\Theta}\sup_{M\in\mathbb{N},\eta_{S_{\rm in|\Pi}S'}}D_h^{\omega}\left[(\Pi(\mathcal{N})\otimes\mathcal{I}_{S'})(\eta_{S_{\rm in|\Pi}S'})\,\middle\|\,\Pi(\mathcal{N})\left(\eta_{S_{\rm in|\Pi}}\right)\otimes\eta_{S'}\right]-\log_2\frac{4\epsilon}{(\epsilon - \omega)^2}\le \sup_{\Pi\in\Theta}C_{(1)}^\epsilon[\Pi(\mathcal{N})] = C_{\Theta,(1)}^\epsilon(\mathcal{N}),
\end{align}
Hence, it suffices to prove the case for the standard scenario with $C_{(1)}$.
\CY{We thus} single out this special case as the following theorem, whose validity implies Theorem~\ref{AppThm:MainResult}. In what follows, $S'$ again denotes an \CY{auxiliary} quantum \CY{system;} $S_{\rm in|\mathcal{N}}$ \CY{and} $S_{\rm out|\mathcal{N}}$ are the input and output systems of the given channel $\mathcal{N}$, respectively.

\begin{atheorem}\label{AppThm:CC} {\em(Bounding One-Shot Classical Capacity)}
Given a channel $\mathcal{N}$ and errors $0<\delta\le\omega<\epsilon\le\CY{1/2}$, we have that
\begin{align}\label{Eq:AppThm:CC}
&\sup_{\substack{M\in\mathbb{N}\\\eta_{S_{\rm in|\mathcal{N}}S'}=\sum_{x=0}^{M-1}p_x\sigma_{x|S_{\rm in|\mathcal{N}}}\otimes\kappa_{x|S'}}}D_h^{\omega}\left[(\mathcal{N}\otimes\mathcal{I}_{S'})(\eta_{S_{\rm in|\mathcal{N}}S'})\,\middle\|\,\mathcal{N}\left(\eta_{S_{\rm in|\mathcal{N}}}\right)\otimes\eta_{S'}\right]-\log_2\frac{4\epsilon}{(\epsilon - \omega)^2}\le\nonumber\\
&\quad\quad\quad\quad\quad\quad  C_{(1)}^\epsilon(\mathcal{N})\le\sup_{\substack{M\in\mathbb{N}\\\mathcal{K}\in{\rm CQ}_{M\to S_{\rm in|\mathcal{N}}}\\\mathcal{L}\in{\rm QC}_{S_{\rm out|\mathcal{N}}\to M}\\\norm{\mathcal{L}\circ\mathcal{N}\circ\mathcal{K}\left(\frac{\id_M}{M}\right) - \frac{\id_M}{M}}_1\le2(\epsilon+\delta)}}D_0^{\epsilon+\delta}\left[((\mathcal{L}\circ\mathcal{N}\circ\mathcal{K})\otimes\mathcal{I}_{M'})(\Phi_{MM'})\,\middle\|\,(\mathcal{L}\circ\mathcal{N}\circ\mathcal{K})\left(\frac{\id_{M}}{M}\right)\otimes\frac{\id_{M'}}{M}\right].
\end{align}
\end{atheorem}
\CY{Before the proof, let us compare Theorem~\ref{AppThm:CC} with the existing bounds 
reported in Ref.~\cite{Wang2013}, which shows that 
$\sup_{M\in\mathbb{N},\sigma_{S_{\rm in|\mathcal{N}}M}}D_h^{\omega}\left[(\mathcal{N}\otimes\mathcal{I}_{M})(\sigma_{S_{\rm in|\mathcal{N}}M})\,\middle\|\,\mathcal{N}\left(\sigma_{S_{\rm in|\mathcal{N}}}\right)\otimes\sigma_{M}\right]$ 
optimising over 
$
\sigma_{S_{\rm in|\mathcal{N}}M} = \sum_{m=0}^{M-1}p_m\rho_{m}\otimes\proj{m}_M
$
can simultaneously upper/lower bound $C_{(1)}^\epsilon(\mathcal{N})$, up to one-shot error terms.
Theorem~\ref{AppThm:CC}'s upper bound implies Ref.~\cite{Wang2013}'s upper bound due to Eq.~\eqref{Eq:useful D_0 D_h},}
\CY{and Theorem~\ref{AppThm:CC}'s lower bound also implies Ref.~\cite{Wang2013}'s lower bound.}
Theorem~\ref{AppThm:MainResult} \CY{thus} generalises \CY{Ref.~\cite{Wang2013}'s results} by extending it to an arbitrary $\Theta$-assisted scenario as well as providing \CY{tighter bounds.}

\CY{
\subsection{Proof of Theorem~\ref{AppThm:CC}}
}
\noindent\CY{\em Proof of the upper bound.}
By definition, we can write $C_{(1)}^\epsilon(\mathcal{N})$ as the following maximisation
\begin{eqnarray}
\begin{aligned}
\max_{M,\mathcal{K},\mathcal{L}}\quad&\log_2M\\
{\rm s.t.}\quad&\frac{1}{M}\sum_{m=0}^{M-1}\bra{m}(\mathcal{L}\circ\mathcal{N}\circ\mathcal{K})(\proj{m})\ket{m}\ge1-\epsilon;\quad\CY{\mathcal{K}\in{\rm CQ}_{M\to S_{\rm in|\mathcal{N}}};\quad\mathcal{L}\in{\rm QC}_{S_{\rm out|\mathcal{N}}\to M};\quad M\in\mathbb{N}.}
\end{aligned}
\end{eqnarray}
\CY{Using Fact~\ref{Fact:A1},} the above optimisation is invariant if we add \CY{the constraint} $\norm{\mathcal{L}\circ\mathcal{N}\circ\mathcal{K}\left(\frac{\id_M}{M}\right) - \frac{\id_M}{M}}_1\le2\epsilon$, which reads
\begin{eqnarray}\label{Eq:max01}
\begin{aligned}
\max_{M,\mathcal{K},\mathcal{L}}\quad&\log_2M\\
{\rm s.t.}\quad&\frac{1}{M}\sum_{m=0}^{M-1}\bra{m}(\mathcal{L}\circ\mathcal{N}\circ\mathcal{K})(\proj{m})\ket{m}\ge1-\epsilon;\\
&\norm{\mathcal{L}\circ\mathcal{N}\circ\mathcal{K}\left(\frac{\id_M}{M}\right) - \frac{\id_M}{M}}_1\le2\epsilon;\quad\CY{\mathcal{K}\in{\rm CQ}_{M\to S_{\rm in|\mathcal{N}}};\quad\mathcal{L}\in{\rm QC}_{S_{\rm out|\mathcal{N}}\to M};\quad M\in\mathbb{N}.}
\end{aligned}
\end{eqnarray}
Now, one can observe that
\begin{align}
{\rm tr}\left[\left(M\Phi_{MM'}\right)\left((\mathcal{L}\circ\mathcal{N}\circ\mathcal{K})\otimes\mathcal{I}_{M'}\right)\left(\Phi_{MM'}\right)\right] &= \frac{1}{M}{\rm tr}\left[\sum_{n=0}^{M-1}\proj{n}_{M}\otimes\proj{n}_{M'}\sum_{m=0}^{M-1}(\mathcal{L}\circ\mathcal{N}\circ\mathcal{K})(\proj{m}_{M})\otimes\proj{m}_{M'}\right]\nonumber\\
& = \frac{1}{M}\sum_{m=0}^{M-1}\bra{m}(\mathcal{L}\circ\mathcal{N}\circ\mathcal{K})(\proj{m})\ket{m}.
\end{align}
On the other hand,
\begin{align}
{\rm tr}\left\{\left(M\Phi_{MM'}\right)\left[(\mathcal{L}\circ\mathcal{N}\circ\mathcal{K})\left(\frac{\id_{M}}{M}\right)\otimes\frac{\id_{M'}}{M}\right]\right\} &= {\rm tr}\left\{\sum_{m=0}^{M-1}\proj{m}_{M}\otimes\proj{m}_{M'}\left[(\mathcal{L}\circ\mathcal{N}\circ\mathcal{K})\left(\frac{\id_{M}}{M}\right)\otimes\frac{\id_{M'}}{M}\right]\right\}\nonumber\\
&= \frac{1}{M}{\rm tr}\left\{\sum_{m=0}^{M-1}\proj{m}_{M}\left[(\mathcal{L}\circ\mathcal{N}\circ\mathcal{K})\left(\frac{\id_{M}}{M}\right)\right]\right\}=\frac{1}{M},
\end{align}
where we have used the fact that $\sum_{m=0}^{M-1}\proj{m}_{M} = \id_{M}$.
\CY{Equation~\eqref{Eq:max01} can thus be rewritten as} 
\begin{eqnarray}
\begin{aligned}\label{Eq:Computation04}
\max_{M,\mathcal{K},\mathcal{L}}\quad&\log_2\frac{1}{{\rm tr}\left\{\left(M\Phi_{MM'}\right)\left[(\mathcal{L}\circ\mathcal{N}\circ\mathcal{K})\left(\frac{\id_{M}}{M}\right)\otimes\frac{\id_{M'}}{M}\right]\right\}}\\
{\rm s.t.}\quad&{\rm tr}\left[\left(M\Phi_{MM'}\right)\left((\mathcal{L}\circ\mathcal{N}\circ\mathcal{K})\otimes\mathcal{I}_{M'}\right)\left(\Phi_{MM'}\right)\right]\ge1-\epsilon;\\
&\norm{\mathcal{L}\circ\mathcal{N}\circ\mathcal{K}\left(\frac{\id_M}{M}\right) - \frac{\id_M}{M}}_1\le2\epsilon;\quad\CY{\mathcal{K}\in{\rm CQ}_{M\to S_{\rm in|\mathcal{N}}};\quad\mathcal{L}\in{\rm QC}_{S_{\rm out|\mathcal{N}}\to M};\quad M\in\mathbb{N}.}
\end{aligned}
\end{eqnarray}
Now, based on our setting, $\mathcal{L}\circ\mathcal{N}\circ\mathcal{K}$ is a $M$-to-$M$ classical channel.
This means the states $\left((\mathcal{L}\circ\mathcal{N}\circ\mathcal{K})\otimes\mathcal{I}_{M'}\right)\left(\Phi_{MM'}\right)$ and $(\mathcal{L}\circ\mathcal{N}\circ\mathcal{K})\left(\frac{\id_{M}}{M}\right)\otimes\frac{\id_{M'}}{M}$ are both diagonal in the basis $\{\ket{n}_M\otimes\ket{m}_{M'}\}_{n,m=0}^{M-1}$.
In other words, one can write
\begin{align}
&\left((\mathcal{L}\circ\mathcal{N}\circ\mathcal{K})\otimes\mathcal{I}_{M'}\right)\left(\Phi_{MM'}\right) = \sum_{n,m}q_{(n,m)}\proj{n}_{M}\otimes\proj{m}_{M'};\\
&(\mathcal{L}\circ\mathcal{N}\circ\mathcal{K})\left(\frac{\id_{M}}{M}\right)\otimes\frac{\id_{M'}}{M} = \sum_{n,m}r_{(n,m)}\proj{n}_{M}\otimes\proj{m}_{M'},
\end{align}
where the weights $q_{(n,m)},r_{(n,m)}$ are given by
\begin{align}\label{AppEq:qnm}
&q_{(n,m)}\coloneqq\bra{n}_M\otimes\bra{m}_{M'}\left((\mathcal{L}\circ\mathcal{N}\circ\mathcal{K})\otimes\mathcal{I}_{M'}\right)\left(\Phi_{MM'}\right)\ket{n}_M\otimes\ket{m}_{M'};\\
&r_{(n,m)}\coloneqq\bra{n}_M\otimes\bra{m}_{M'}\left[(\mathcal{L}\circ\mathcal{N}\circ\mathcal{K})\left(\frac{\id_{M}}{M}\right)\otimes\frac{\id_{M'}}{M}\right]\ket{n}_M\otimes\ket{m}_{M'}.\label{AppEq:rnm}
\end{align}
Then Eq.~\eqref{Eq:Computation04} can be rewritten as
\begin{eqnarray}
\begin{aligned}
\max_{M,\mathcal{K},\mathcal{L}}\quad&\log_2\frac{1}{\sum_{n=m}r_{(n,m)}}\\
{\rm s.t.}\quad&\CY{\sum_{n=m}q_{(n,m)}\ge1-\epsilon;\quad\norm{\mathcal{L}\circ\mathcal{N}\circ\mathcal{K}\left(\frac{\id_M}{M}\right) - \frac{\id_M}{M}}_1\le2\epsilon;\quad\mathcal{K}\in{\rm CQ}_{M\to S_{\rm in|\mathcal{N}}};\quad\mathcal{L}\in{\rm QC}_{S_{\rm out|\mathcal{N}}\to M};\quad M\in\mathbb{N},}
\end{aligned}
\end{eqnarray}
which is upper bounded by
\begin{eqnarray}
\begin{aligned}\label{Eq:Computation05}
\max_{M,\mathcal{K},\mathcal{L},\Lambda}\quad&\log_2\frac{1}{\sum_{(n,m)\in\Lambda}r_{(n,m)}}\\
{\rm s.t.}\quad&\CY{\sum_{(n,m)\in\Lambda}q_{(n,m)}\ge1-\epsilon;\quad\Lambda:{\rm a\;subset\;of\;indices}\;(n,m);}\\
&\CY{\norm{\mathcal{L}\circ\mathcal{N}\circ\mathcal{K}\left(\frac{\id_M}{M}\right) - \frac{\id_M}{M}}_1\le2\epsilon;\quad\mathcal{K}\in{\rm CQ}_{M\to S_{\rm in|\mathcal{N}}};\quad\mathcal{L}\in{\rm QC}_{S_{\rm out|\mathcal{N}}\to M};\quad M\in\mathbb{N}.}
\end{aligned}
\end{eqnarray}
\CY{Using} 
Eq.~\eqref{AppEq:D0epsilon}, we can further upper bound Eq.~\eqref{Eq:Computation05} by the following for every $0<\delta<\epsilon<\CY{1/2}$:
\begin{eqnarray}
\begin{aligned}\label{Eq:Computation06}
\max_{M,\mathcal{K},\mathcal{L}}\quad&D_0^{\epsilon+\delta}\left[\left((\mathcal{L}\circ\mathcal{N}\circ\mathcal{K})\otimes\mathcal{I}_{M'}\right)\left(\Phi_{MM'}\right)\,\middle\|\,(\mathcal{L}\circ\mathcal{N}\circ\mathcal{K})\left(\frac{\id_{M}}{M}\right)\otimes\frac{\id_{M'}}{M}\right]\\
{\rm s.t.}\quad&\CY{\norm{\mathcal{L}\circ\mathcal{N}\circ\mathcal{K}\left(\frac{\id_M}{M}\right) - \frac{\id_M}{M}}_1\le2\epsilon;\quad\mathcal{K}\in{\rm CQ}_{M\to S_{\rm in|\mathcal{N}}};\quad\mathcal{L}\in{\rm QC}_{S_{\rm out|\mathcal{N}}\to M};\quad M\in\mathbb{N}.}
\end{aligned}
\end{eqnarray}
This implies the desired upper bound by extending $2\epsilon$ to $2(\epsilon+\delta)$ in the constraint.
\CY{\hfill$\square$}\\

\noindent\CY{\em Proof of the lower bound.}
The proof follows the same strategy adopted in Ref.~\cite{Wang2013}.
For the completeness of this work, we still provide a detailed proof.
First, we recall the following inequality from Lemma 2 in Ref.~\cite{Hayashi2003}:\\

\begin{alemma}{\em(Hayashi-Nagaoka Inequality~\cite{Hayashi2003})}\label{AppLemma:HayashiLemma}
For every $0\le A\le\id$, $B\ge0$, and positive value $c>0$, we have that
\begin{align}
\id - (A+B)^{-\frac{1}{2}}A(A+B)^{-\frac{1}{2}}\le (1+c)(\id-A)+(2+c+c^{-1})B.
\end{align}
\end{alemma}
For the rest of the proof, we let $S_{\rm in} = S_{\rm in|\mathcal{N}}$ and $S_{\rm out} = S_{\rm out|\mathcal{N}}$ with the given channel $\mathcal{N}$.
To begin with, let us fix errors $0<\omega<\epsilon$, a given $M\in\mathbb{N}$, a finite dimensional \CY{auxiliary} system $S'$ (which is {\em not necessary} of the same size with $S_{\rm in}$), a separable state of the form $\eta_{S_{\rm in}S'} = \sum_{x=0}^{M-1}p_x\sigma_{x|S_{\rm in}}\otimes\kappa_{x|S'}$ with a probability distribution $\{p_x\}_{x=0}^{M-1}$ and states $\sigma_{x|S_{\rm in}}$, $\kappa_{x|S'}$, and an operator $0\le Q\le\id_{S_{\rm out}S'}$ satisfying 
\begin{align}\label{AppEq:LowerBoundAssumption}
{\rm tr}\left[Q(\mathcal{N}\otimes\mathcal{I}_{S'})(\eta_{S_{\rm in}S'})\right]\ge1-\omega.
\end{align}
Now, for a given $L\in\mathbb{N}$ and a {\em codebook} $\mathcal{C} = \{m_i\}_{i=1}^{L}$ (namely, it is a mapping, $i\mapsto m_i$, from the set of classical information $\{i\}_{i=1}^L$ with size $L$ to the set $\{m\}_{m=0}^{M-1}$), we consider the encoding and decoding scheme for the classical messages of the size $L$ given by $\{\rho_{m_i}\}_{i=1}^L$ and $\{E_{i|\mathcal{C}}\}_{i=1}^L$, which are defined by:
\CY{\begin{align}
\rho_{x}\coloneqq\sigma_{x|S_{\rm in}};\quad\quad\quad E_{i|\mathcal{C}}\coloneqq\left(\sum_{k=1}^LA_{m_k}\right)^{-\frac{1}{2}}A_{m_i}\left(\sum_{k=1}^LA_{m_k}\right)^{-\frac{1}{2}}\quad{\rm if}\;i>1;\quad\quad\quad E_{1|\mathcal{C}}\coloneqq\id_{{\rm S}_{\rm out}} - \sum_{k=2}^LE_{k|\mathcal{C}},
\end{align}}
where
\begin{align}
A_{x}\coloneqq{\rm tr}_{S'}\left[\left(\id_{S_{\rm out}}\otimes\kappa_{x|S'}\right)Q\right].
\end{align}
Note that for a normal operator $A$, we define $A^{-1}$ to be the inverse of $A$ in its support; namely, $A^{-1}A = AA^{-1}$ will be the projection onto its support (which is the space ${\rm span}\{\ket{\phi_i}\}_{i=1}^{N_A}$, where $\ket{\phi_i}$'s are eigenstates of $A$ with non-zero eigenvalue satisfying $\braket{\phi_i}{\phi_j} = \delta_{ij}$, and $N_A$ is the number of non-zero eigenvalues that $A$ has, including degeneracy).
This means that in general we have $AA^{-1}\le\id$, and thus
\begin{align}\label{AppEq:ComputationDetail001}
E_{1|\mathcal{C}}& = \left(\sum_{k=1}^LA_{m_k}\right)^{-\frac{1}{2}}A_{m_1}\left(\sum_{k=1}^LA_{m_k}\right)^{-\frac{1}{2}} + \id_{{\rm S}_{\rm out}} - \sum_{l=1}^L\left(\sum_{k=1}^LA_{m_k}\right)^{-\frac{1}{2}}A_{m_l}\left(\sum_{k=1}^LA_{m_k}\right)^{-\frac{1}{2}}\CY{\ge\left(\sum_{k=1}^LA_{m_k}\right)^{-\frac{1}{2}}A_{m_1}\left(\sum_{k=1}^LA_{m_k}\right)^{-\frac{1}{2}}\ge0.}
\end{align}
Hence, we learn that $\sum_{i=1}^LE_{i|\mathcal{C}} = \id_{{\rm S}_{\rm out}}$ and $E_{i|\mathcal{C}}\ge0$ for every $i$; i.e., $\{E_{i|\mathcal{C}}\}_{i=1}^L$ is a POVM in the system ${\rm S}_{\rm out}$.

Now, we define the {\em average failure probability} to be one minus the average success probability (recall from Definition~\ref{AppDef:CCapacity}).
For a given codebook $\mathcal{C}$ and the given channel $\mathcal{N}$, the failure probability of transmitting {\em $L$ bits of} classical information, when using the corresponding encoding $\{\rho_{m_i}\}_{i=1}^L$ and decoding $\{E_{i|\mathcal{C}}\}_{i=1}^L$, is given by
\begin{align}\label{AppEq:PfailDef}
P_{\rm fail}(\mathcal{N},\mathcal{C}) \coloneqq \frac{1}{L}\sum_{i=1}^LP_{\rm fail}(\mathcal{N},\mathcal{C}|m_i)\coloneqq\frac{1}{L}\sum_{i=1}^L{\rm tr}\left[(\id_{S_{\rm out}} - E_{i|\mathcal{C}})\mathcal{N}(\rho_{m_i})\right]=1-\frac{1}{L}\sum_{i=1}^L{\rm tr}\left[E_{i|\mathcal{C}}\mathcal{N}(\rho_{m_i})\right],
\end{align}
where 
$
P_{\rm fail}(\mathcal{N},\mathcal{C}|m_i)\coloneqq{\rm tr}\left[(\id_{S_{\rm out}} - E_{i|\mathcal{C}})\mathcal{N}(\rho_{m_i})\right]=1-{\rm tr}\left[E_{i|\mathcal{C}}\mathcal{N}(\rho_{m_i})\right]
$
is the probability to {\em wrongly} decode the $i$th output.
Now we compute the failure probability {\em averaging over all possible codebooks} according to the probability distribution $\{p_x\}_{x=0}^{M-1}$ associated with the state $\eta_{S_{\rm in}S'}$.
To this end, in what follows, we will treat each $m_i$ as a random variable that draws an element from $\{m\}_{m=0}^{M-1}$ based on the probability distribution $\{p_m\}_{m=0}^{M-1}$.
To compute the average, we adopt the notation 
\begin{align}
\CY{\mathbb{E}_{x}f(x)\coloneqq\sum_{x=0}^{M-1}p_xf(x)}
\end{align} 
for every function $f(x)$ in $x$ (note that this notation depends on the given probability distribution $\{p_x\}_{x=0}^{M-1}$, while we keep this dependency implicit for simplicity).
Then the average reads
\begin{align}
\left(\prod_{i=1}^L\mathbb{E}_{m_i}\right)P_{\rm fail}(\mathcal{N},\mathcal{C}) = \left(\prod_{i=1}^L\mathbb{E}_{m_i}\right)\sum_{k=1}^L\frac{1}{L}P_{\rm fail}(\mathcal{N},\mathcal{C}|m_k)=\sum_{k=1}^L\frac{1}{L}\mathbb{E}_{m_k}\left(\prod_{\substack{i=1\\i\neq k}}^L\mathbb{E}_{m_i}\right)P_{\rm fail}(\mathcal{N},\mathcal{C}|m_k).
\end{align}
Now we note that, for every $k=1,...,L$, using Lemma~\ref{AppLemma:HayashiLemma} with $A = A_{m_k}$ and $B = \sum_{\substack{i=1\\i\neq k}}^LA_{m_i}$ gives
\begin{align}
\id_{S_{\rm out}} - E_{k|\mathcal{C}}\le(1+c)\left(\id_{S_{\rm out}}-A_{m_k}\right)+(2+c+c^{-1})\sum_{\substack{i=1\\i\neq k}}^LA_{m_i},
\end{align}
which holds for every fixed $c>0$, and we will optimise over $c$ in the end.
Note that for $k=1$ we have used the lower bound in Eq.~\eqref{AppEq:ComputationDetail001}.
Hence, we have, for every $k=1,...,L$,
\begin{align}
P_{\rm fail}(\mathcal{N},\mathcal{C}|m_k)&={\rm tr}\left[(\id_{S_{\rm out}} - E_{k|\mathcal{C}})\mathcal{N}(\rho_{m_k})\right]\CY{\le(1+c){\rm tr}\left[\left(\id_{S_{\rm out}} - A_{m_k}\right)\mathcal{N}(\rho_{m_k})\right] + (2+c+c^{-1}){\rm tr}\left[\left(\sum_{\substack{i=1\\i\neq k}}^LA_{m_i}\right)\mathcal{N}(\rho_{m_k})\right].}
\end{align}
Now we note that both $\rho_{m_k}$ and $A_{m_k}$ only depend on $m_k$, rather than the whole codebook $\mathcal{C}$.
This means 
\begin{align}
\mathbb{E}_{m_k}\left(\prod_{\substack{i=1\\i\neq k}}^L\mathbb{E}_{m_i}\right)P_{\rm fail}(\mathcal{N},\mathcal{C}|m_k)&\le(1+c)\left[1-\mathbb{E}_{m_k}{\rm tr}\left(A_{m_k}\mathcal{N}(\rho_{m_k})\right)\right] + (2+c+c^{-1}){\rm tr}\left[\left(\mathbb{E}_{m_k}\mathcal{N}(\rho_{m_k})\right)\left(\prod_{\substack{i=1\\i\neq k}}^L\mathbb{E}_{m_i}\sum_{\substack{i=1\\i\neq k}}^LA_{m_i}\right)\right]\nonumber\\
&\le(1+c)\left[1-\mathbb{E}_{x}{\rm tr}\left(A_{x}\mathcal{N}(\rho_{x})\right)\right] + (2+c+c^{-1}){\rm tr}\left[\left(\mathbb{E}_{x}\mathcal{N}(\rho_{x})\right)\left(\sum_{\substack{i=1\\i\neq k}}^L\mathbb{E}_{m_i}A_{m_i}\right)\right]\nonumber\\
&=(1+c)\left[1-\mathbb{E}_{x}{\rm tr}\left(A_{x}\mathcal{N}(\rho_{x})\right)\right] + (2+c+c^{-1}){\rm tr}\left[\left(\mathbb{E}_{x}\mathcal{N}(\rho_{x})\right)\times(L-1)\times\left(\mathbb{E}_{x}A_{x}\right)\right].
\end{align}
Following Ref.~\cite{Wang2013}, we note that, 
\begin{align}
\mathbb{E}_{x}{\rm tr}\left[A_{x}\mathcal{N}(\rho_{x})\right] = \sum_{x=0}^{M-1}p_x{\rm tr}\left\{{\rm tr}_{S'}\left[\left(\id_{S_{\rm out}}\otimes\kappa_{x|S'}\right)Q\right]\mathcal{N}(\sigma_{x|S_{\rm in}})\right\}={\rm tr}\left[Q(\mathcal{N}\otimes\mathcal{I}_{S'})(\eta_{S_{\rm in}S'})\right]\ge1-\omega,
\end{align}
Which is due to our assumption Eq.~\eqref{AppEq:LowerBoundAssumption}.
On the other hand,
\begin{align}
{\rm tr}\left[\left(\mathbb{E}_{x}A_{x}\right)\times\mathbb{E}_{x}\mathcal{N}(\rho_{x})\right]={\rm tr}\left\{{\rm tr}_{S'}\left[\left(\id_{S_{\rm out}}\otimes\sum_{x=0}^{M-1}p_x\kappa_{x|S'}\right)Q\right]\mathcal{N}\left(\sum_{y=0}^{M-1}p_y\sigma_{y|S'}\right)\right\}={\rm tr}\left[Q\left(\mathcal{N}\left(\eta_{S_{\rm in}}\right)\otimes\eta_{S'}\right)\right].
\end{align}
Combining everything, we learn that
\begin{align}\label{AppEq:NecessaryConditionUpperBound}
\left(\prod_{i=1}^L\mathbb{E}_{m_i}\right)P_{\rm fail}(\mathcal{N},\mathcal{C})=\sum_{k=1}^L\frac{1}{L}\mathbb{E}_{m_k}\left(\prod_{\substack{i=1\\i\neq k}}^L\mathbb{E}_{m_i}\right)P_{\rm fail}(\mathcal{N},\mathcal{C}|m_k)\le(1+c)\omega+(2+c+c^{-1})(L-1){\rm tr}\left[Q\left(\mathcal{N}\left(\eta_{S_{\rm in}}\right)\otimes\eta_{S'}\right)\right].
\end{align}
From here we observe that if Eq.~\eqref{AppEq:NecessaryConditionUpperBound} is further upper bounded by $\epsilon$, then $\left(\prod_{i=1}^L\mathbb{E}_{m_i}\right)P_{\rm fail}(\mathcal{N},\mathcal{C})\le\epsilon$, meaning that there must exist at least one codebook, denoted by $\mathcal{C}_L$, with certain combinations of $m_i$'s such that $P_{\rm fail}(\mathcal{N},\mathcal{C}_L)\le\epsilon$.
For this codebook, the corresponding encoding $\{\rho_{m_i}\}_{i=1}^L$ and decoding $\{E_{i|\mathcal{C}_L}\}_{i=1}^L$ form a feasible solution to the maximisation of $C_{(1)}^\epsilon(\mathcal{N})$ given in Definition~\ref{AppDef:CCapacity} (with $\Theta=\Theta_{\rm C}$).
Hence, we learn that [see also Eq.~\eqref{AppEq:PfailDef}]
\begin{align}
(1+c)\omega+(2+c+c^{-1})(L-1){\rm tr}\left[Q\left(\mathcal{N}\left(\eta_{S_{\rm in}}\right)\otimes\eta_{S'}\right)\right]\le\epsilon\quad\Rightarrow\quad\log_2L\le C_{(1)}^\epsilon(\mathcal{N}).
\end{align}
But since we know that no such feasible solution can exist when $L=2^{C_{(1)}^\epsilon(\mathcal{N})}+1$, we conclude that $\epsilon$ must be {\em upper bounded} by the upper bound given in Eq.~\eqref{AppEq:NecessaryConditionUpperBound} when $L=2^{C_{(1)}^\epsilon(\mathcal{N})}+1$.
This implies that
\begin{align}
\epsilon<(1+c)\omega+(2+c+c^{-1})2^{C_{(1)}^\epsilon(\mathcal{N})}\times{\rm tr}\left[Q\left(\mathcal{N}\left(\eta_{S_{\rm in}}\right)\otimes\eta_{S'}\right)\right].
\end{align}
Since this argument works for every $M\in\mathbb{N}$, $\eta_{S_{\rm in}S'}$ (with some finite dimensional \CY{auxiliary} system $S'$), and $0\le Q\le\id_{S_{\rm out}S'}$ satisfying ${\rm tr}\left[Q\left(\mathcal{N}\otimes\mathcal{I}_{S'}\right)\left(\eta_{S_{\rm in}S'}\right)\right]\ge1-\omega$, we conclude that, when $\epsilon>(1+c)\omega$,
\begin{align}\label{AppEq:ForAsymptoticProof}
\sup_{\substack{M\in\mathbb{N},\eta_{S_{\rm in}S'},0\le Q\le\id_{S_{\rm out}S'}\\{\rm tr}\left[Q\left(\mathcal{N}\otimes\mathcal{I}_{S'}\right)\left(\eta_{S_{\rm in}S'}\right)\right]\ge1-\omega}}\log_2\frac{1}{{\rm tr}\left[Q\left(\mathcal{N}\left(\eta_{S_{\rm in}}\right)\otimes\eta_{S'}\right)\right]}-\log_2\frac{2+c+c^{-1}}{\epsilon - (1+c)\omega}\le C_{(1)}^\epsilon(\mathcal{N}).
\end{align}
Using the definition Eq.~\eqref{AppEq:HypothesisTesting}, we have
\begin{align}
\sup_{M\in\mathbb{N},\eta_{S_{\rm in}S'}}D_h^{\omega}\left[(\mathcal{N}\otimes\mathcal{I}_{S'})(\eta_{S_{\rm in}S'})\,\middle\|\,\mathcal{N}\left(\eta_{S_{\rm in}}\right)\otimes\eta_{S'}\right]-\log_2\frac{2+c+c^{-1}}{\epsilon - (1+c)\omega}\le C_{(1)}^\epsilon(\mathcal{N}).
\end{align}
Now, by choosing $c = \frac{\epsilon - \omega}{\epsilon+\omega}$, we have, when $\epsilon>\omega>0$,  
\begin{align}
\epsilon-(1+c)\omega = \epsilon-\frac{2\epsilon\omega}{\epsilon+\omega}=\epsilon c,
\end{align}
This shows that we indeed have $(1+c)\omega<\epsilon$ with this $c$ value.
Finally, direct computation shows that
\begin{align}
\frac{2+c+c^{-1}}{\epsilon - (1+c)\omega} = \frac{1}{\epsilon(\epsilon-\omega)^2}\left[2\left(\epsilon^2 - \omega^2\right)+\left(\epsilon - \omega\right)^2+\left(\epsilon + \omega\right)^2\right]=\frac{4\epsilon}{(\epsilon - \omega)^2},
\end{align}
\CY{showing the desired lower bound.
\hfill$\square$}
\end{widetext}
\CY{Theorem~\ref{AppThm:MainResult} provides a general way to quantitatively describe one-shot $\Theta$-assisted classical capacities via entropic quantities.
These entropic quantities take key roles in bridging communication and energy transmission---as detailed in the next section, they can also characterise energy transmission tasks.
}

\CY{
\section{Energy Transmission via Channels
}\label{Sec:energy transmission}
}

\subsection{\AA berg's $\epsilon$-Deterministic Work Extraction}\label{App:Proof-MainResult-WorkExtraction}
\CY{To start with,} we briefly \CY{recap} \AA berg's formulation on one-shot work extraction~\cite{Aberg2013}. 
Consider a fixed energy \CY{eigenbasis}, denoted by $\{\ket{n}\}_{n=0}^{N-1}$, and a fixed temperature $T$.
Within this setup, all possible system Hamiltonians are of the form $H = \sum_{n=0}^{N-1}E_n\proj{n}$, and only \CY{\em energy-incoherent} states, i.e., states of the form $\eta = \sum_{n=0}^{N-1}\eta_n\proj{n}$ (which means that $[\eta,H]=0$) are considered.
Such a state can be equivalently characterised by a random variable $\proj{n}\mapsto n$ outputting the value $n$ with probability $\eta_n$.
For a given Hamiltonian $H = \sum_{n=0}^{N-1}E_n\proj{n}$, measuring energy in the state $\eta$ again gives a random \CY{variable}
\begin{align}\label{Eq:measuring energy}
\CY{
H_\eta:\proj{n}\mapsto E_n
}
\end{align}
with 
\begin{align}\label{Eq:measuring energy prob}
\CY{{\rm P}\left(\left\{H_\eta = E_n\right\}\right) = \eta_n,}
\end{align}
where $P(\{\cdot\})$ denotes the probability of the give event $\{\cdot\}$ to happen, and $\left\{H_\eta = E_n\right\}$ denotes the event that the energy is evaluated in the eigenstate $\ket{n}$ of $\eta$
\CY{(see Supplementary Note 2 and Eq.~(S1) in Ref.~\cite{Aberg2013})}.
Using $(\eta,H)$ to jointly indicate the system's state and Hamiltonian, 
\CY{we have:} 
\begin{adefinition}
{\em ({\AA berg's Work Extraction Scenario}~\cite{Aberg2013})
For an energy-incoherent state $\eta$, a {\em work extraction process of $\eta$ subject to Hamiltonian $H$} is a mapping that brings the pair $(\eta,H)$ to \CY{another} pair $(\cdot,H)$ with the same Hamiltonian that is composed by finitely many \CY{``allowed operations''}~\footnote{
\CY{Not to be confused with the allowed operations of quantum resource theories. Note that the thermalisation defined here is a valid channel and thus satisfies the golden rule of resource theories Eq.~\eqref{Eq:golden rule}.
However, level transformations are {\em not} channels at all, and thus cannot be captured by Eq.~\eqref{Eq:golden rule}.}
}. 
\CY{Here, the allowed operations are:}
\begin{enumerate}
\item ({\em Level Transformation}) One is allowed to change the Hamiltonian's energy levels.
Such an operation takes the form 
\begin{align}
\CY{(\rho,H)\mapsto(\rho,H'), }
\end{align}
where $H=\sum_{n=0}^{N-1}E_n\proj{n},H'=\sum_{n=0}^{N-1}E'_n\proj{n}$ are spanned by the same eigenbasis while with different (finite) energy gaps.
To realise this operation, one needs to tune the Hamiltonian promptly so that the system's state remains unchanged.
This can be interpreted as an isentropic process \CY{(a {\em quench} operation)}.
Importantly, for $\rho = \sum_{n=0}^{N-1}\rho_n\proj{n}$, one define the {\em work cost} as the following random variable during this operation:
\begin{align}
\CY{W = (H'_\rho - H_\rho):\proj{n}\mapsto E'_n - E_n}
\end{align}
with 
\begin{align}
\CY{{\rm P}\left(\left\{H'_\rho - H_\rho = E'_n - E_n\right\}\right) = \rho_n.}
\end{align}
The quantity $-W$ is the {\em extractable work}.
\item ({\em Thermalisation}) One is allowed to thermalise the system, which is the mapping 
\begin{align}
\CY{(\rho,H)\mapsto\left(\gamma_H,H\right),}
\end{align}
\CY{where $\gamma_H$ is the thermal state defined in Eq.~\eqref{Eq: thermal state}.}
Physically, it means that the system is in contact with a large bath with temperature $T$ and achieves thermal equilibrium.
During this operation, the Hamiltonian is invariant.
Also, we assume that there is no work cost associated with this operation.
\end{enumerate}
}
\end{adefinition}
In general, since one can combine two level transformations into one, and so do two thermalisation processes, the {\em total work cost} (which is again a random variable) of a work extraction process $\mathcal{P}$ of the given state-Hamiltonian pair $(\eta,H)$ is given by
\CY{
(see also Supplementary Note 2 of Ref.~\cite{Aberg2013})}
\begin{align}
W(\mathcal{P},\eta,H)&\coloneqq H^{(1)}_\eta - H_\eta + \sum_{k=1}^{\CY{K-1}}\left(H^{(k+1)}_{\gamma_{H^{(k)}}} - H^{(k)}_{\gamma_{H^{(k)}}}\right)\nonumber\\
&\quad+ H_{\gamma_{H^{(K)}}} - H^{(K)}_{\gamma_{H^{(K)}}}.
\end{align}
That is, it consists of $K+1$ level transformations changing Hamiltonians sequentially as 
\begin{align}
\CY{H\mapsto H^{(1)}\mapsto H^{(2)}\mapsto ...\mapsto H^{(K)}\mapsto H,}
\end{align}
\CY{and a thermalisation is inserted between every two level transformations.}
\CY{Here, 
$H^{(k')}_{\gamma_{H^{(k)}}}$ is the random variable of measuring energy in $\gamma_{H^{(k)}}$ with the Hamiltonian $H^{(k')}$ [Eqs.~\eqref{Eq:measuring energy} and~\eqref{Eq:measuring energy prob}].}
Note that there is no need to add an additional thermalisation in the end since it contributes zero work cost.
Denote by $\mathfrak{P}(\eta,H)$ all work extraction processes of $\eta$ subject to $H$, then one can define the {\em $(\epsilon,\delta)$-deterministic extractable work of $\eta$ subject to $H$} with $0<\epsilon<1$ and \mbox{$0\le\delta<\infty$} as the following quantity \CY{(which is a combination of Supplementary Definitions 4, 7, and 8 in Ref.~\cite{Aberg2013})}
\begin{align}\label{AppEq:Wext}
&W_{\rm ext,(1)}^{(\epsilon,\delta)}\left(\eta,H\right)\coloneqq\nonumber\\
&-\inf\bigcup_{\mathcal{P}\in\mathfrak{P}(\eta,H)}\left\{w\in\mathbb{R}\,\middle|\,P\left[\left\{\left|W(\mathcal{P},\eta,H) - w\right|\le\delta\right\}\right]>1-\epsilon\right\}.
\end{align}
$W_{\rm ext,(1)}^{(\epsilon,\delta)}\left(\eta,H\right)$ is the highest value which the work gain can be ``$\delta$-closed to'' with a probability no less than $1-\epsilon$.
By requesting an arbitrarily good precision $\delta\to0$, \CY{we obtain} 
the {\em one-shot $\epsilon$-deterministic extractable work}, which is given by
\begin{align}
W_{\rm ext,(1)}^{\epsilon}\left(\eta,H\right)\coloneqq\lim_{\delta\to0}W_{\rm ext,(1)}^{(\epsilon,\delta)}\left(\eta,H\right),
\end{align}
\CY{which} is an one-shot extractable work with high predictability.
\CY{The main theorem of Ref.~\cite{Aberg2013} (whose formal statement is given in Supplementary Corollary 1 in its supplementary information)} can then be summarised as follows:
\begin{atheorem}\label{AppThm:Aberg}{\em(\AA berg's One-Shot Work Extraction Theorem~\cite{Aberg2013})}
For an energy-incoherent state $\eta$ with Hamiltonian $H$, a temperature $0<T<\infty$, and an error $0<\epsilon\le1-\CY{1/\sqrt{2}}$, we have that 
\begin{align}
0\le W_{\rm ext,(1)}^\epsilon(\eta,H) - (k_BT\ln2)D_0^\epsilon\left(\eta\,\|\,\gamma_H\right)\le k_BT\ln\frac{1}{1-\epsilon}.
\end{align}
\end{atheorem}
Note that the entropy $D_0^\epsilon$ \CY{defined in Eq.~\eqref{AppEq:D0epsilon}} is understood to be computed in the energy eigenbasis, and the thermal state $\gamma_H$ is \CY{again defined by Eq.~\eqref{Eq: thermal state}.}
For the complete consideration and the detailed framework, we refer the reader to Ref.~\cite{Aberg2013} \CY{(especially Supplementary Notes 2, 8, and 9)}.

\CY{
\subsection{Work Extraction from Correlation}
After knowing how to quantify extractable work from states, we now analyse the extractable work from states' {\em correlation}.
Consider a given bipartite state $\rho_{AB}$.
Following Ref.~\cite{Perarnau-Llobet2015}'s approach, this can be achieved by preparing local Hamiltonians (with global Hamiltonian of the form $H_A\otimes\id_B + \id_A\otimes H_B$) so that $\rho_{AB}$ is {\em locally thermal}---its local states  $\rho_A = \gamma_{H_A}$ and $\rho_B = \gamma_{H_B}$ are thermal states [Eq.~\eqref{Eq: thermal state}] to the local Hamiltonians $H_A$ and $H_B$ (here, again, we assume a fixed background temperature $0<T<\infty$ is given).
Since no work can be extracted from thermal equilibrium (Theorem~\ref{AppThm:Aberg}), extracted work, if there is any, must come from global correlation.
In other words, we identify the extractable work from $\rho_{AB}$'s correlation as the one extracted from $\rho_{AB}$ when local Hamiltonians make it locally thermal.

To analytically characterise the extractable work from correlation via Theorem~\ref{AppThm:Aberg}, we focus on separable $\rho_{AB}$ satisfying \begin{align}\label{Eq:simultaneously diagonalisable}
[\rho_{AB},\rho_A\otimes\rho_B]=0.
\end{align}
Physically, this means that when $\rho_{AB}$ is locally thermal, it is also energy-incoherent~\footnote{\CY{More precisely, when the bipartite Hamiltonian $H_A\otimes\id_B+\id_A\otimes H_B$ makes $\rho_{AB}$ locally thermal, the bipartite thermal state is \mbox{$\rho_A\otimes\rho_B$}. Then, Eq.~\eqref{Eq:simultaneously diagonalisable} means that $\rho_{AB}$ and this bipartite thermal state are simultaneously diagonalisable.
Hence, there is a bipartite energy eigenbasis simultaneously diagonalising both $\rho_{AB}$ and $\rho_A\otimes\rho_B$, meaning that $\rho_{AB}$ is energy-incoherent to this energy eigenbasis.
We can then use this energy eigenbasis to apply \AA berg's work extraction scenario and Theorem~\ref{AppThm:Aberg}.}}.
Then, one can use Theorem~\ref{AppThm:Aberg} to bound the {\em one-shot $\epsilon$-deterministic extractable work from $\rho_{AB}$'s correlation}, denoted by $W_{\rm corr,(1)}^\epsilon(\rho_{AB})$: 
For every \mbox{$0<T<\infty$} and $0<\epsilon\le1-\CY{1/\sqrt{2}}$, we have [see Fig.~\ref{Fig:Wcorr_tasks} (a)]
\begin{align}\label{AppEq:App:Proof-Com}
0&\le W_{\rm corr,(1)}^\epsilon(\rho_{AB}) - (k_BT\ln2)D_0^\epsilon(\rho_{AB}\,\|\,\rho_A\otimes\rho_B)\nonumber\\
&\le k_BT\ln\frac{1}{1-\epsilon}.
\end{align}
These bounds will play a key role in characterising the energy transmission tasks that we are about to introduce.}

\begin{figure*}
\begin{center}
\scalebox{0.8}{\includegraphics{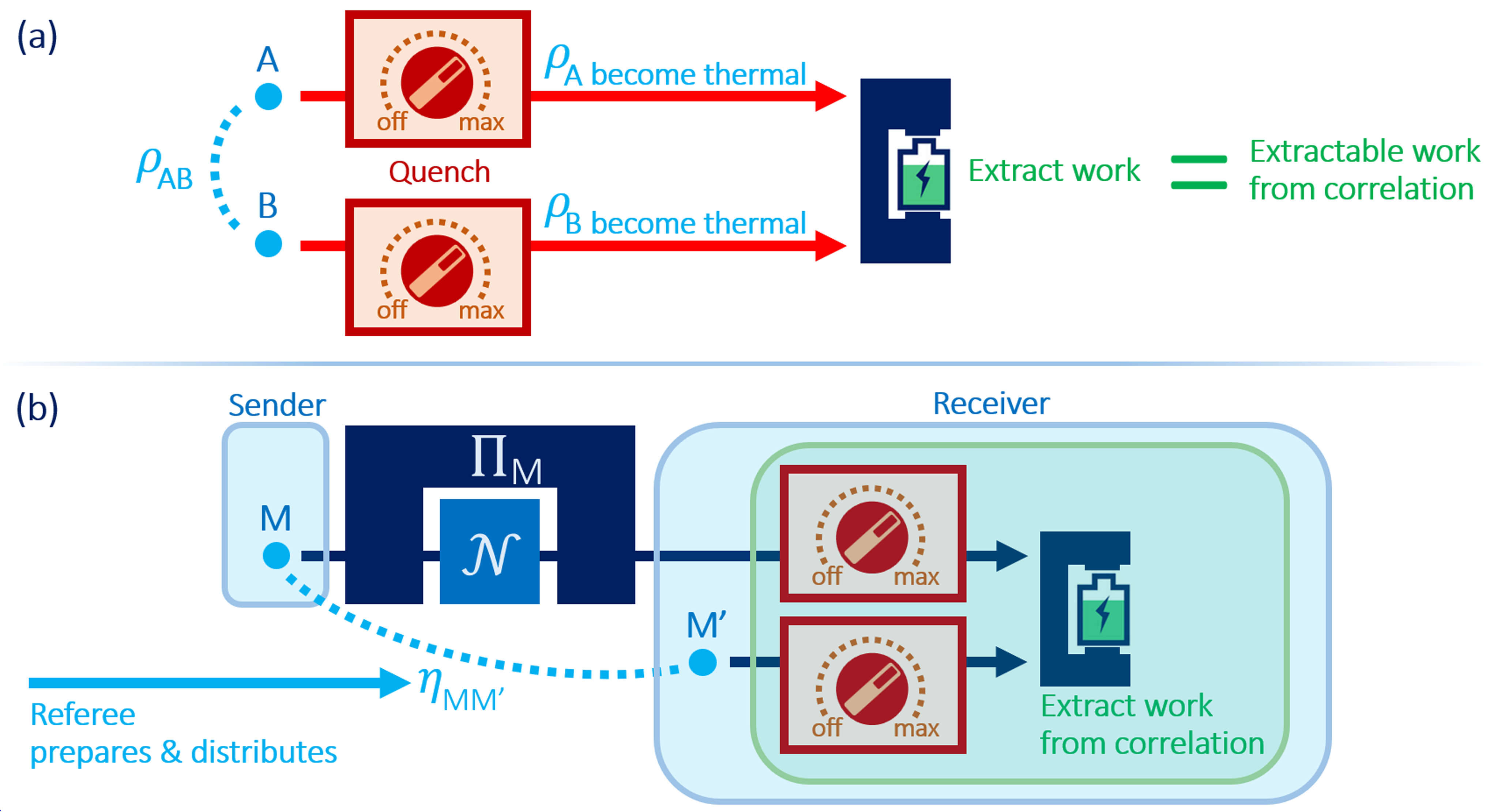}}
\caption{
\CY{\bf Tasks for extracting and transmitting work-like energy.}
\CY{(a) The task corresponding to Eq.~\eqref{AppEq:App:Proof-Com}, which aims to extract work from the correlation of a given state $\rho_{AB}$.
This can be done by first quenching the local Hamiltonians to make the state locally thermal (i.e., making both $\rho_A$ and $\rho_B$ the thermal states of the local systems $A$ and $B$) and then extracting work from the bipartite state.
\CYnew{(b) The task corresponding to Eq.~\eqref{Eq: one-shot transmitted energy}, which aims to measure the work-like energy that can {\em only} due to transmission via the channel $\mathcal{N}$ (with the assistance of $\Pi_M\in\Theta_M$).
As argued in the text, the extracted work from the bipartite output's correlation can only result from transmission by $\Pi_M(\mathcal{N})$.}}
}
\label{Fig:Wcorr_tasks} 
\end{center}
\end{figure*}

\CYnew{
\subsection{Energy Transmission Tasks}
Now, we introduce an operational task to analyse the work-like energy {\em definitely transmitted} by a channel $\mathcal{N}$, as detailed below [see also Fig.~\ref{Fig:Wcorr_tasks} (b)].
Consider a setting with referee, sender, and receiver. 
At the beginning, the referee prepares a bipartite state $\eta_{MM'}$ diagonal in a given computational basis \mbox{$\{\ket{n}_M\otimes\ket{m}_{M'}\}_{n,m=0}^{M-1}$} subject to some  initial Hamiltonian of the form $H_M\otimes\id_{M'} + \id_M\otimes H_{M'}$, where $\{\ket{n}_M\}_{n=0}^{M-1}$ ($\{\ket{m}_{M'}\}_{m=0}^{M-1}$) is an energy eigenbasis of $H_M$ ($H_{M'}$).
\CYthree{We further assume $H_M$ and $H_{M'}$ are of finite-energy.}

In the next step, the referee distributes the part $M$ ($M'$) of $\eta_{MM'}$ to the sender (receiver).
Then, the sender uses a classical version of $\mathcal{N}$, denoted by $\Pi_M(\mathcal{N})$, to locally send $\eta_{MM'}$'s part $M$ to the receiver [recall that $\Pi_M(\mathcal{N})$ is an $M$-to-$M$ classical channel induced by $\mathcal{N}$ and $\Theta$; see Definition~\ref{AppDef:ClassicalVersion}].
We demand that $\Pi_M(\mathcal{N})$'s output system is again a classical system described by the energy eigenbasis $\{\ket{n}_M\}_{n=0}^{M-1}$.

After completing the transmission step as mentioned above, the receiver possesses a bipartite state. 
Now, our goal is to identify the amount of energy that can {\em only} result from {\em transmission} rather than being {\em created} by the channel $\Pi_M(\mathcal{N})$.
For instance, if $\Pi_M(\mathcal{N})$ is a so-called erasure channel acting as $(\cdot)\mapsto\proj{0}$, then $\Pi_M(\mathcal{N})$ will {\em create} extractable energy which is not transmitted.
This also suggests that, in general, the extractable work from output's part $M$ is not solely resulting from transmission.
On the other hand, the extractable work from output's part $M'$ is also not due to transmission at all, since it was with the receiver {\em before} applying $\Pi_M(\mathcal{N})$. 

Hence, to isolate the receiver's energy gain that is {\em definitely} due to energy transmission, we check the extractable work from the receiver's output {\em bipartite correlation}.
This correlation can be quantified by $D_0^\epsilon$ via
\begin{align}
D_0^\epsilon\left[(\Pi_M(\mathcal{N})\otimes\mathcal{I}_{M'})(\eta_{MM'})\,\|\,\Pi_M(\mathcal{N})(\eta_M)\otimes\eta_{M'}\right],
\end{align}
which is a smoothed version of the min-mutual information~\cite{Datta2013}.
Since $\Pi_M(\mathcal{N})$ only locally acts on $\eta_{MM'}$'s part $M$, it cannot generate any bipartite correlation---a fact quantitatively captured by the data-processing inequality of $D_0^\epsilon$ for a small enough $\epsilon$ (see Fact~\ref{App: data-processing ineq D_0} in Appendix~\ref{App}):
\begin{align}\label{Eq: must be transmitted 01}
&D_0^\epsilon\left[(\Pi_M(\mathcal{N})\otimes\mathcal{I}_{M'})(\eta_{MM'})\,\|\,\Pi_M(\mathcal{N})(\eta_M)\otimes\eta_{M'}\right]\nonumber\\
&\quad\le D_0^\epsilon\left(\eta_{MM'}\,\|\,\eta_M\otimes\eta_{M'}\right).
\end{align}
Namely, the output bipartite correlation is upper bounded by the input bipartite correlation, i.e., $D_0^\epsilon\left(\eta_{MM'}\,\|\,\eta_M\otimes\eta_{M'}\right)$.
Together with Eq.~\eqref{AppEq:App:Proof-Com},
we thus obtain, for a small enough $\epsilon$,
\begin{align}\label{Eq: must be transmitted 02}
&W_{\rm corr,(1)}^\epsilon [(\Pi_M(\mathcal{N})\otimes\mathcal{I}_{M'})(\eta_{MM'})]\nonumber\\
&\quad\quad\le W_{\rm corr,(1)}^\epsilon(\eta_{MM'}) + k_BT\ln\frac{1}{1-\epsilon}.
\end{align}
Thus, the local channel $\Pi_M(\mathcal{N})$ {\em cannot generate} extractable work from global correlation.
Namely, work extracted from receiver's bipartite correlation can only be {\em preserved} or {\em maintained} by the channel $\Pi_M(\mathcal{N})$---it can only result from {\em transmission}.
We then measure the highest energy that is {\em definitely} transmitted by $\mathcal{N}$ by the {\em one-shot $\Theta$-assisted $\epsilon$-deterministic genuinely transmitted energy}
defined as
\begin{align}\label{Eq: one-shot transmitted energy}
&W^{\epsilon}_{\rm corr|\Theta,(1)}(\mathcal{N})\coloneqq\nonumber\\
&\quad\sup_{\substack{M\in\mathbb{N},\eta_{MM'}\\\Pi_M\in\Theta_M}}W_{\rm corr,(1)}^\epsilon\left[(\Pi_M(\mathcal{N})\otimes\mathcal{I}_{M'})(\eta_{MM'})\right]
\end{align}

Note that, importantly, for the validity of Eq.~\eqref{Eq: one-shot transmitted energy}, we need to check Eq.~\eqref{AppEq:App:Proof-Com} is indeed applicable.
Namely, we need to show that $(\Pi_M(\mathcal{N})\otimes\mathcal{I}_{M'})(\eta_{MM'})$ can satisfy Eq.~\eqref{Eq:simultaneously diagonalisable}. To see this, since $\eta_{MM'}$ is diagonal in the given computational basis \mbox{$\{\proj{n}_M\otimes\proj{m}_{M'}\}_{n,m=0}^{M-1}$,} we can write 
\mbox{$
\eta_{MM'} = \sum_{n,m}p_{nm}\proj{n}_M\otimes\proj{m}_{M'}. 
$}
Also, by the setting, $\Pi_M(\mathcal{N})$'s 
output is always diagonal in the given basis $\{\ket{n}_M\}_{n=0}^{M-1}$.
Hence, for every $n$, we can write
\mbox{$
\Pi_M(\mathcal{N})(\proj{n}_M) = \sum_{l}p'_{l|n}\proj{l}_M.
$}
Combining everything and defining $q_{lm}\coloneqq\sum_np'_{l|n}p_{nm}$, we obtain
\begin{align}
(\Pi_M(\mathcal{N})\otimes\mathcal{I}_{M'})(\eta_{MM'})= \sum_{l,m}q_{lm}\proj{l}_M\otimes\proj{m}_M,
\end{align}
which is indeed simultaneously diagonalisable with \mbox{$\Pi_M(\mathcal{N})(\eta_{M})\otimes\eta_{M'}$} in the given computational basis \mbox{$\{\proj{n}_M\otimes\proj{m}_{M'}\}_{n,m=0}^{M-1}$}.
Hence, Eq.~\eqref{Eq:simultaneously diagonalisable} is satisfied, and this computational basis acts as the fixed energy eigenbasis for \AA berg's work extraction scenario and Theorem~\ref{AppThm:Aberg}.}

\CY{
Finally, we consider a simplified task by imposing three additional constraints on the task defining Eq.~\eqref{Eq: one-shot transmitted energy}:
\begin{enumerate}
\item Initial Hamiltonians $H_M,H_{M'}$ are fully degenerate.
\item $\eta_{MM'}=\Phi_{MM'}$ is maximally correlated [Eq.~\eqref{Eq: max corr state}].
\item $\norm{\Pi_M(\mathcal{N})\left(\id_M/M\right) - \id_M/M}_1<2\epsilon$.
\end{enumerate}
Here, $\id_M/M$ describes thermal equilibrium when the system Hamiltonian is fully degenerate.
Hence, the third condition means we only allowed classical versions to generate informational non-equilibrium up to the order $O(\epsilon)$.
This task induces the following highest transmitted energy:
\begin{align}\label{Eq:Wcorr alternative def}
&W_{\rm \Phi|\Theta,(1)}^\epsilon(\mathcal{N})\coloneqq\nonumber\\
&\sup_{\substack{M\in\mathbb{N},\Pi_M\in\Theta_M\\\norm{\Pi_M(\mathcal{N})\left(\frac{\id_M}{M}\right) - \frac{\id_M}{M}}_1<2\epsilon}}W_{\rm corr,(1)}^\epsilon\left[(\Pi_M(\mathcal{N})\otimes\mathcal{I}_{M'})(\Phi_{MM'})\right].
\end{align}
In the next section, we will use the figure-of-merits introduced here to bridge the transmissions of information and energy.}

\CY{
\section{Thermodynamic Bounds on Classical Capacities
}\label{Sec:TCTCI}
}

\CY{
\subsection{Bounding One-Shot $\Theta$-Assisted Classical Capacities by Energy Transmission Tasks}\label{App:Proof}
}
\CY{We now present this work's first major result, which bridges the one-shot transmissions of information and energy (this is Theorem~1 in the companion paper~\cite{Companion2}):}\\

\begin{atheorem}{\em(Quantifying Classical Communication by Work Extraction~\cite{Companion2})}~\label{Result:TCTCI}
%
\CY{Consider} a set of superchannels $\Theta$ and a fixed temperature $0<T<\infty$.
For a channel $\mathcal{N}$ and errors $0<\delta\le\omega<\epsilon\le1-\CY{1/\sqrt{2}}$, we have that
\begin{align}
&W_{\rm corr|\Theta,(1)}^{\omega}(\mathcal{N})-k_BT\ln\frac{4\epsilon}{(\epsilon - \omega)^2(1-\omega)}\nonumber\\
&\quad\quad\le(k_BT\ln2)C_{\Theta,(1)}^\epsilon(\mathcal{N})\le W_{\rm \Phi|\Theta,(1)}^{\epsilon+\delta}(\mathcal{N}).
\end{align}
\end{atheorem}
\begin{proof}
\CY{First, Eq.~\eqref{Eq:Wcorr alternative def} and Theorem~\ref{AppThm:MainResult} implies}
the upper bound.
To show the lower bound, \CY{note that} when $\eta_{MM'}$ is
\CY{diagonal in the given computational basis $\{\proj{n}_M\otimes\proj{m}_{M'}\}_{n,m=0}^{M-1}$,} 
one can always write 
\begin{align}\label{Eq: proof needed 001}
\CY{\eta_{MM'} = \sum_{m=0}^{M-1}p_m\sigma_{m|M}\otimes\kappa_{m|M'}}
\end{align}
with states $\sigma_{m|M}$ \CY{and} $\kappa_{m|M'}$.
\CY{Then we have}
\begin{widetext}
\begin{align}
&\sup_{\substack{M\in\mathbb{N},\Pi\in\Theta\\\eta_{S_{\rm in|\Pi}S'}=\sum_{x=0}^{M-1}p_x\sigma_{x|S_{\rm in|\Pi}}\otimes\kappa_{x|S'}}}D_h^{\omega}\left[(\Pi(\mathcal{N})\otimes\mathcal{I}_{S'})(\eta_{S_{\rm in|\Pi}S'})\,\middle\|\,\Pi(\mathcal{N})\left(\eta_{S_{\rm in|\Pi}}\right)\otimes\eta_{S'}\right]\nonumber\\
&\quad\quad\quad\ge\sup_{\substack{M\in\mathbb{N},\Pi\in\Theta\\\mathcal{K}\in{\rm CQ}_{M\to S_{\rm in|\Pi}}\\\eta_{MM'}:\;{\rm classical\;in\;}MM'}}D_h^{\omega}\left[(\Pi(\mathcal{N})\circ\mathcal{K}\otimes\mathcal{I}_{M'})(\eta_{MM'})\,\middle\|\,\Pi(\mathcal{N})\circ\mathcal{K}\left(\eta_{M}\right)\otimes\eta_{M'}\right]\nonumber\\
&\quad\quad\quad\ge\sup_{\substack{M\in\mathbb{N},\Pi\in\Theta\\\mathcal{K}\in{\rm CQ}_{M\to S_{\rm in|\Pi}}\\\mathcal{L}\in{\rm QC}_{S_{\rm out|\Pi}\to M}\\\eta_{MM'}:\;{\rm classical\;in\;}MM'}}D_h^{\omega}\left[(\mathcal{L}\circ\Pi(\mathcal{N})\circ\mathcal{K}\otimes\mathcal{I}_{M'})(\eta_{MM'})\,\middle\|\,\mathcal{L}\circ\Pi(\mathcal{N})\circ\mathcal{K}\left(\eta_{M}\right)\otimes\eta_{M'}\right]\nonumber\\
&\quad\quad\quad\ge\sup_{\substack{M\in\mathbb{N},\Pi_M\in\Theta_M\\\eta_{MM'}:\;{\rm classical\;in\;}MM'}}D_0^{\omega}\left[(\Pi_M(\mathcal{N})\otimes\mathcal{I}_{M'})(\eta_{MM'})\,\middle\|\,\Pi_M(\mathcal{N})\left(\eta_{M}\right)\otimes\eta_{M'}\right]\nonumber\\
&\quad\quad\quad\ge \frac{W_{\rm corr|\Theta,(1)}^\omega(\mathcal{N})}{k_BT\ln2} - \log_2\frac{1}{1-\omega}.
\end{align}
\end{widetext}
The first line is from 
\CY{Theorem~\ref{AppThm:MainResult}'s lower bound}, which is further lower bounded by the second line \CY{by restricting} the maximisation range \CY{via Eq.~\eqref{Eq: proof needed 001}.}
The third line is \CY{from $D_h^\omega$'s data-processing inequality, and the fourth line is due to Eq.~\eqref{Eq:useful D_0 D_h} and Definition~\ref{AppDef:ClassicalVersion}.}
\CY{The} last line follows from 
\CY{Eqs.~\eqref{Eq: one-shot transmitted energy} and~\eqref{AppEq:App:Proof-Com}, and}
the proof is completed by using Theorem~\ref{AppThm:MainResult} again.
\end{proof}
\CY{
Note that the term $k_BT\ln[4\epsilon/(\epsilon - \omega)^2(1-\omega)]$ cannot change the physical meaning.
To see this, consider $k$ copies of the channel $\mathcal{N}$; i.e., $\mathcal{N}^{\otimes k}$.
When $k\to\infty$, we have
\begin{align}\label{Eq: one-shot error meaning}
\frac{W_{\rm corr|\Theta,(1)}^{\omega}\left(\mathcal{N}^{\otimes k}\right)}{k}&\approx\frac{(k_BT\ln2)C_{\Theta,(1)}^\epsilon\left(\mathcal{N}^{\otimes k}\right)}{k}\nonumber\\
&\approx\frac{W_{\rm \Phi|\Theta,(1)}^{\epsilon+\delta}\left(\mathcal{N}^{\otimes k}\right)}{k}.
\end{align}
Namely, the contribution from \mbox{$k_BT\ln[4\epsilon/(\epsilon - \omega)^2(1-\omega)]$} vanishes when we consider sufficiently many (but still finite) copies of the channel ($\mathcal{N}^{\otimes k}$) and then take the average over the copy number ($k$).
This is the reason why a term like $k_BT\ln[4\epsilon/(\epsilon - \omega)^2(1-\omega)]$ is called an ``one-shot error term''---they cannot change the physics, especially in the asymptotic (i.e., many copies) regime.
Hence, we conclude that $W_{\rm corr|\Theta,(1)}$, $(k_BT\ln2)C_{\Theta,(1)}$, and $W_{\rm \Phi|\Theta,(1)}$ carry the {\em same} and {\em equivalent} physical meaning, and Theorem~\ref{Result:TCTCI} provides a quantitative equivalence between transmitting information and energy.}
\CY{In the next section, we discuss the physical implications of Theorem~\ref{Result:TCTCI}.
Especially, Theorem~\ref{Result:TCTCI} enables us to uncover a dynamical version of Landauer's principle~\cite{Landauer1961}.}

\CYnew{
\subsection{Dynamical Landauer's Principle}\label{Sec:Landauer}
}
\CYnew{As a surprising physical implication, Theorem~\ref{Result:TCTCI} enables us to uncover a {\em dynamical version} of Landauer's principle~\cite{Landauer1961}.
Loosely speaking, for a qubit with background temperature $T$, Landauer's principle states that preparing a pure state $\ket{\psi}$ from the maximally mixed one, i.e., $\ket{\psi}\mapsto\id/2$, must be accompanied by at least $k_BT\ln2$ energy cost.   
Together with Szilard engine~\cite{Szilard1929} (see, e.g., Section IV in Ref.~\cite{Hsieh2021}), one can equate informational and energetic properties of states: A state carries one bit of deterministic information (e.g., a qubit pure state) {\em if and only if} it possesses one unit of extractable work ($k_BT\ln2$).
These are two {\em static} properties of quantum states, which can be viewed as states' information and energy content.
Armed with Theorem~\ref{Result:TCTCI}, we can now equate informational and energetic properties of quantum channels, which are their abilities to transmit information and energy---we can equate two {\em dynamical} properties of quantum channels.}

\CYnew{
We start with the (one-shot) $\Theta$-assisted scenario (as in Fig.~\ref{Fig:CCtasks}).
A channel $\mathcal{N}$ can be viewed
from two different perspectives---information-theoretical and thermodynamic.
when we view it information-theoretically as a communication channel (Fig.~\ref{Fig:CCtasks}), it can transmit $n$ bits of (classical) information if and only if $n\le C_{\Theta,(1)}^\epsilon(\mathcal{N})$, as defined in Definition~\ref{AppDef:CCapacity}.
Now, in exactly the same scenario, by setting appropriate Hamiltonians, we can also view it thermodynamically as an energy-transmitting process as in Eqs.~\eqref{Eq: one-shot transmitted energy} and~\eqref{Eq:Wcorr alternative def}. 
Theorem~\ref{Result:TCTCI} then implies that {\em necessarily and sufficiently}, 
the channel $\mathcal{N}$ must possess the ability to transmit $n\times(k_BT\ln2)$ amount of energy, up to one-shot error terms:}\\

\CYnew{
\begin{acorollary}\label{coro:weak dynamical Landauer}
The ability to transmit $n$ bits of information is equivalent to the ability to transmit $n\times(k_BT\ln2)$ energy.
\end{acorollary}

Crucially, the equivalence between the {\em abilities} to do two things {\em may not} always imply these two things are equivalent. 
To see the equivalence between transmitting information and energy, we now argue that they can happen {\em simultaneously}. In fact, they can be two facets of the {\em same} physical process.
To this end, we present an explicit scenario in which 
(1) transmitting information and energy happen simultaneously, and
(2) information transmission must be accompanied by energy transmission.
This can thus be viewed as the dynamical version of Landauer's principle.

Consider a channel $\mathcal{N}$ with \mbox{$C_{\Theta,(1)}^\epsilon(\mathcal{N})\ge n\coloneqq\log_2M$} in the energy transmission task defining via Eq.~\eqref{Eq: one-shot transmitted energy} ($\epsilon>0$ is given and sufficiently small). 
Namely, it is a channel having the ability to transmit $n$ bits of information.
In the task, suppose the referee prepares the maximally correlated state \mbox{$\eta_{MM'}=\Phi_{MM'} = \frac{1}{M}\sum_{m=0}^{M-1}\proj{m}_{M}\otimes\proj{m}_{M'}$} as the bipartite input with fully degenerate Hamiltonians.
Physically, $\Phi_{MM'}$ is prepared as a statistical mixture of the product state $\ket{m}_M\otimes\ket{m}_{M'}$ in a multi-trials experiment as follows: During each trial, with a uniformly distributed probability $1/M$, the referee prepares a pair of pure states $\ket{m}_M$ and $\ket{m}_{M'}$.
After that, the referee sends $\ket{m}_M$ to the sender and sends $\ket{m}_{M'}$ to the receiver, both with fully degenerate Hamiltonians.
We now argue that, in this setting, when $n = \log_2M$ bits of information are transmitted, it {\em must} be accompanied by at least $n\times(k_BT\ln2)$ transmitted energy.}

\CYnew{
First, using Fact~\ref{AppLemma:AlternativeCCapacity}, transmitting $n$ bits of information in the current setting means that we apply a classical version \mbox{$\Pi_M\in\Theta_M$} satisfying, for every $m=0,...,M-1$,
\begin{align}\label{Eq:dynamical Landauer comp001}
\norm{\Pi_M(\mathcal{N})(\proj{m})-\proj{m}}_1=O(\epsilon).
\end{align}
Here, ``$O(\epsilon)$'' denotes a term satisfying $\lim_{\epsilon\to0}O(\epsilon)=0$, and $\norm{\cdot}_1$ is the trace norm defined in Eq.~\eqref{Eq:trace norm}.
In most trials, by applying $\Pi_M(\mathcal{N})$, the receiver's bipartite output will be very close to the form \mbox{$\ket{m}_M\otimes\ket{m}_{M'}$} for some $m$.
That is, $\ket{m}_M$ can be sent from the sender to the receiver {\em almost} reliably, up to an error of the order $O(\epsilon)$.
This means that $\mathcal{N}$ is not just having the ``ability'' to transmit information---it is indeed doing so via the physical process $\Pi_M(\mathcal{N})$. 
We thus conclude:
\begin{center}
{\em {\bf (Observation A)} $n$ bits of information are transmitted by the physical process described by $\Pi_M(\mathcal{N})$.}
\end{center}}

\CYnew{
Let us argue that work-like energy is {\em also} transmitted by the {\em same} physical process $\Pi_M(\mathcal{N})$.
The first thing to note is that the receiver {\em does not} have any net extractable work-like energy if $\Pi_M(\mathcal{N})$ {\em has not} been applied.
This is because $M'$ is the {\em only} system the receiver can have in the absence of $\Pi_M(\mathcal{N})$.
when the receiver obtains $\ket{m}_{M'}$ from the referee, after multiple trials, it is statistically described by $\id_{M'}/M$ (as $\ket{m}_{M'}$'s are uniformly distributed).
Since it is with a fully degenerate Hamiltonian, no (net) work can be extracted. 
Namely,
\begin{center}
{\em{\bf(Observation B)} Without the physical process $\Pi_M(\mathcal{N})$, the receiver possesses no net work-like energy gain.}
\end{center}

Crucially, the situation changes when one applies the physical process $\Pi_M(\mathcal{N})$---with sufficiently many trials, the receiver's {\em bipartite} output is described by the statistical mixture 
$
[\Pi_M(\mathcal{N})\otimes\mathcal{I}_{M'}](\Phi_{MM'})
$ with the condition
\begin{align}
\norm{[\Pi_M(\mathcal{N})\otimes\mathcal{I}_{M'}](\Phi_{MM'})-\Phi_{MM'}}_1=O(\epsilon),
\end{align}
where Eq.~\eqref{Eq:dynamical Landauer comp001} has been used.
{\em After} the physical process $\Pi_M(\mathcal{N})$, the receiver can add additional work extraction protocols to extract work from this bipartite output's correlation.
Using Eq.~\eqref{AppEq:App:Proof-Com}, extractable work from $\Phi_{MM'}$'s correlation is \mbox{$W_{\rm corr,(1)}^\epsilon(\Phi_{MM'})=n\times(k_BT\ln2)+O(\epsilon)$.}
Due to the continuity of $D_0^\epsilon$ (see Fact~\ref{App:coro} in Appendix~\ref{App}), we conclude that, when $\epsilon$ is small enough, the extractable work from $[\Pi_M(\mathcal{N})\otimes\mathcal{I}_{M'}](\Phi_{MM'})$'s correlation is $n\times(k_BT\ln2)$, up to an error of the order $O(\epsilon)$~\footnote{\CYnew{Note that Fact~\ref{App:coro} is applicable since $\Pi_M(\mathcal{N})$ is a classical-to-classical channel, and the four states $[\Pi_M(\mathcal{N})\otimes\mathcal{I}_{M'}](\Phi_{MM'})$, \mbox{$\Pi_M(\mathcal{N})(\id_M/M)\otimes\id_{M'}/M$}, $\Phi_{MM'}$, and $\id_M/M\otimes\id_{M'}/M$ are {\em all} diagonalised in the given computational basis $\{\ket{n}_M\otimes\ket{m}_{M'}\}_{n,m=0}^{M-1}$.
}}.
This is the work-like energy {\em brought by} the physical process $\Pi_M(\mathcal{N})$, since the receiver does not have any extractable work in the absence of $\Pi_M(\mathcal{N})$ (Observation B).
In other words,
\begin{center}
{\em{\bf(Observation C)} After the physical process $\Pi_M(\mathcal{N})$, the receiver can have $n\times(k_BT\ln2)$ amounts of work gain.}
\end{center}
Finally, as we argue before [in particular, Eqs.~\eqref{Eq: must be transmitted 01} and~\eqref{Eq: must be transmitted 02}], this amount of work-like energy can {\em only} result from {\em transmission}, since it {\em cannot be created} by the physical process $\Pi_M(\mathcal{N})$.
Consequently, we obtain
\begin{center}
{\em{\bf(Observation D)} 
At least $n\times(k_BT\ln2)$ amounts of work-like energy has been transmitted to the receiver by $\Pi_M(\mathcal{N})$.
}
\end{center}

Note that Observations A and D mean that $n$ bits of information and $n\times(k_BT\ln2)$ amounts of work-like energy are {\em both} transmitted by $\Pi_M(\mathcal{N})$.
Still, to {\em extract} these transmitted resources, the receiver needs to implement additional processes---just like when we know the packages have been shipped and arrived, we still need to open the mailbox to ``extract'' them.
In a multi-trial experiment, the receiver can do so by dividing the trials into two batches.
For the first batch, they apply decoding to extract information; for the second batch, they apply a work extraction scenario to extract work.}

\CYnew{
Combining Observations A, B, C, and D, we thus conclude that the channel $\mathcal{N}$
is transmitting information {\em and} work-like energy {\em at the same time} via the physical process $\Pi_M(\mathcal{N})$.
Moreover, the amount of transmitted work-like energy is {\em guaranteed}, or even {\em demanded}, by the amount of transmitted information.
We thus obtain a truly work-like, genuinely dynamical version of Landauer's principle:
\begin{acorollary}\label{coro:strong dynamical Landauer}
In the above setting, transmitting $n$ bits of classical information must be accompanied by transmitting $n\times(k_BT\ln2)$ amounts of work-like energy.
\end{acorollary}
Interestingly, since the physical process $\Pi_M(\mathcal{N})$ is solely processing information, 
the energy transmission discussed here is {\em mediated} by information transmission. 
This can thus be viewed as a dynamical counterpart of the famous information-to-work conversion via Szilard engine~\cite{Szilard1929}.
We have thus fully answered this work's central question.}

\CY{
\subsection{Asymptotic Limit and Holevo-Schumacher-Westmoreland Theorem}\label{App:AsymptoticLimit}
}
\CY{So far, we have fully addressed the one-shot cases.
By applying the asymptotic limit (namely, by considering multiple copies of the given channel and then averaging over the copy number), our results can reproduce} the well-known \CY{\em Holevo-Schumacher-Westmoreland} (HSW) Theorem~\cite{Holevo1973,Holevo1998,Schumacher1997}, which describes the exact form of the standard classical capacity (here we use the form given by Theorem 20.3.1 in Ref.~\cite{Wilde-book}).
First, the {\em Holevo information} of a channel $\mathcal{N}$ with input space $S_{\rm in|\mathcal{N}} = S_{\rm in}$ is defined by~\cite{Holevo1973,Holevo1998} (see also Ref.~\cite{Wilde-book})
\begin{align}
\chi(\mathcal{N})\coloneqq\max_{\sigma_{S_{\rm in}M}}S\left[(\mathcal{N}\otimes\mathcal{I}_M)\left(\sigma_{S_{\rm in}M}\right)\,\|\,\mathcal{N}\left(\sigma_{S_{\rm in}}\right)\otimes\sigma_M\right],
\end{align}
where the maximisation is taken over every \CY{\em quantum-classical state} $\sigma_{S_{\rm in}M}\coloneqq\sum_{m=0}^{M-1}p_m\rho_{m}\otimes\proj{m}_M$ with some classical system $M$, and recall that the {\em quantum relative entropy} is defined by $S(\rho\,\|\,\sigma)\coloneqq{\rm tr}\left[\rho\left(\log_2\rho - \log_2\sigma\right)\right]$.
\CY{$\chi(\mathcal{N})$} measures \CY{$\mathcal{N}$'s ability}
to maintain the classical correlation between sender and receiver.
In fact, we have that (see Ref.~\cite{Wilde-book} for a thorough introduction)\\

\begin{atheorem}\label{AppThm:HSW}{\em(Holevo-Schumacher-Westmoreland's Classical Capacity Theorem~\cite{Holevo1998,Schumacher1997})}
For every channel $\mathcal{N}$, \CY{we have}
\begin{align}\label{Eq:Asymptotic C def}
C(\mathcal{N})\coloneqq\lim_{\epsilon\to0}\varliminf_{k\to\infty}\frac{1}{k}C_{(1)}^\epsilon\left(\mathcal{N}^{\otimes k}\right) = \lim_{k\to\infty}\frac{1}{k}\chi\left(\mathcal{N}^{\otimes k}\right).
\end{align}
\end{atheorem}
\CY{Note that $\varliminf$ and $\varlimsup$ are the so-called {\em limit inferior} and {\em limit superior}, respectively, which are notions generalising the usual limit (see, e.g., Ref.~\cite{Analysis-book}).}
Its proof can be included in the following \CY{{\em restricted version}:}
\begin{widetext}
\begin{aproposition}\label{AppResult:HSWThermo}{\em(Holevo-Schumacher-Westmoreland's Theorem with Thermodynamic Constraints)}
For every channel $\mathcal{N}$ and $0<\theta<\CY{1/2}$, we have that
\begin{align}
C(\mathcal{N}) = \lim_{k\to\infty}\frac{1}{k}\overline{\chi}^\theta\left(\mathcal{N}^{\otimes k}\right),
\end{align}
where
\begin{align}
\overline{\chi}^\theta(\mathcal{N})\coloneqq\sup_{\substack{M\in\mathbb{N}\\\mathcal{K}\in{\rm CQ}_{M\to S_{{\rm in}|\mathcal{N}}}\\\mathcal{L}\in{\rm QC}_{S_{{\rm out}|\mathcal{N}}\to M}\\\norm{\mathcal{L}\circ\mathcal{N}\circ\mathcal{K}\left(\frac{\id_M}{M}\right) - \frac{\id_M}{M}}_1\le2\theta}}S\left[\left(\left(\mathcal{L}\circ\mathcal{N}\circ\mathcal{K}\right)\otimes\mathcal{I}_{M'}\right)(\Phi_{MM'})\,\middle\|\,\left(\mathcal{L}\circ\mathcal{N}\circ\mathcal{K}\right)\left(\frac{\id_{M}}{M}\right)\otimes\frac{\id_{M'}}{M}\right]
\end{align}
\end{aproposition}
\CY{We remark that $\overline{\chi}^\theta$ defined above can be understood as Hovelo information with a thermodynamic constraint. Namely, with fully degenerate Hamiltonians, the encoding and decoding cannot drive the system out of thermal equilibrium too much.
In other words, the abilities of encoding and decoding to generate informational non-equilibrium is limited by the parameter $\theta$.}
\begin{proof}
We first prove the upper bound.
Theorem~\ref{AppThm:CC} implies that, for every $0<\delta<\epsilon\le\CY{1/2}$, 
\begin{align}
C_{(1)}^\epsilon(\mathcal{N}^{\otimes k})&\le\sup_{\substack{M\in\mathbb{N}\\\mathcal{K}\in{\rm CQ}_{M\to S_{{\rm in}|\mathcal{N}^{\otimes k}}}\\\mathcal{L}\in{\rm QC}_{S_{{\rm out}|\mathcal{N}^{\otimes k}}\to M}\\\norm{\mathcal{L}\circ\mathcal{N}^{\otimes k}\circ\mathcal{K}\left(\frac{\id_M}{M}\right) - \frac{\id_M}{M}}_1\le2(\epsilon+\delta)}}D_h^{\epsilon+\delta}\left[\left(\left(\mathcal{L}\circ\mathcal{N}^{\otimes k}\circ\mathcal{K}\right)\otimes\mathcal{I}_{M'}\right)(\Phi_{MM'})\,\middle\|\,\left(\mathcal{L}\circ\mathcal{N}^{\otimes k}\circ\mathcal{K}\right)\left(\frac{\id_{M}}{M}\right)\otimes\frac{\id_{M'}}{M}\right]\nonumber\\
&\le\frac{H_b(\epsilon+\delta)}{1-\epsilon-\delta}+\frac{1}{1-\epsilon-\delta}\overline{\chi}^{\epsilon+\delta}\left(\mathcal{N}^{\otimes k}\right),
\end{align}
where we have used the fact that $D_0^\epsilon(\rho\,\|\,\sigma)\le D_h^\epsilon(\rho\,\|\,\sigma)$, whenever they are well-defined, and the following estimate~\cite{Wang2013}:
\begin{align}
D_h^\epsilon(\rho\,\|\,\sigma)\le\frac{1}{1-\epsilon}\left[S(\rho\,\|\,\sigma) + H_b(\epsilon)\right]
\end{align}
with the binary entropy function $H_b(x)\coloneqq-x\log_2x - (1-x)\log_2(1-x)$ with $H_b(0)\coloneqq0$.
This means that, for every $0<\theta<\CY{1/2}$, by focusing on $\epsilon$'s values small enough so that $\epsilon+\delta\le\theta$,
\begin{align}\label{AppEq:HSWUpperBound}
C(\mathcal{N}) \coloneqq \lim_{\epsilon\to0}\varliminf_{k\to\infty}\frac{1}{k}C_{(1)}^\epsilon\left(\mathcal{N}^{\otimes k}\right) \le \lim_{\epsilon\to0}\varliminf_{k\to\infty}\frac{1}{1-\epsilon-\delta}\frac{1}{k}\overline{\chi}^\theta\left(\mathcal{N}^{\otimes k}\right) = \varliminf_{k\to\infty}\frac{1}{k}\overline{\chi}^\theta\left(\mathcal{N}^{\otimes k}\right)\le \varliminf_{k\to\infty}\frac{1}{k}\chi\left(\mathcal{N}^{\otimes k}\right),
\end{align}
which is the desired bound.

Now it remains to show that $\varlimsup_{k\to\infty}\frac{1}{k}\chi\left(\mathcal{N}^{\otimes k}\right)\le C(\mathcal{N})$ to complete the proof, which has been shown by Ref.~\cite{Wang2013}.
For the completeness of this work, we still detail the proof here.
Using the lower bound of Theorem~\ref{AppThm:CC} (which implies Wang-Renner's direct coding bound~\cite{Wang2013}), we learn that, for a fixed $k\in\mathbb{N}$ (recall that $S_{\rm in} = S_{\rm in|\mathcal{N}}$)
\begin{align}
C_{(1)}^\epsilon\left(\mathcal{N}^{\otimes kl}\right)+\log_2\frac{4\epsilon}{(\epsilon - \omega)^2}&\ge\max_{M\in\mathbb{N},\sigma_{S_{\rm in}^{\otimes kl}M}}D_h^{\omega}\left[\left(\mathcal{N}^{\otimes kl}\otimes\mathcal{I}_{M}\right)\left(\sigma_{S_{\rm in}^{\otimes kl}M}\right)\,\middle\|\,\mathcal{N}^{\otimes kl}\left(\sigma_{S_{\rm in}^{\otimes kl}}\right)\otimes\sigma_{M}\right]\nonumber\\
&\ge\max_{N\in\mathbb{N},\sigma_{S_{\rm in}^{\otimes k}N}}D_h^{\omega}\left[\left(\left(\mathcal{N}^{\otimes k}\otimes\mathcal{I}_{N}\right)\left(\sigma_{S_{\rm in}^{\otimes k}N}\right)\right)^{\otimes l}\,\middle\|\,\left(\mathcal{N}^{\otimes k}\left(\sigma_{S_{\rm in}^{\otimes k}}\right)\otimes\sigma_{N}\right)^{\otimes l}\right],
\end{align}
which holds for every $l\in\mathbb{N}$.
Now, we recall the following lemma~\cite{Hiai1991,Ogawa2000} 
\begin{alemma}\label{AppLemma:QStein}{\em(Quantum Stein's Lemma~\cite{Hiai1991,Ogawa2000})}
For every states $\rho,\sigma$ with ${\rm supp}(\rho)\subseteq{\rm supp}(\sigma)$ and $0<\epsilon<1$, we have
\begin{align}
\lim_{n\to\infty}\frac{1}{n}D_h^\epsilon\left(\rho^{\otimes n}\,\|\,\sigma^{\otimes n}\right) = S(\rho\,\|\,\sigma).
\end{align}
\end{alemma}
Note that for every positive integers $n_1,n_2\in\mathbb{N}$, we have that $C_{(1)}^\epsilon\left(\mathcal{N}^{\otimes n_1}\right)\le C_{(1)}^\epsilon\left(\mathcal{N}^{\otimes (n_1+n_2)}\right)$. 
This means, for every $N\in\mathbb{N},\sigma_{S_{\rm in}^{\otimes k}N}$ with the fixed $k$, we have that (note that $\omega<\epsilon$, and the notation $[x]$ denotes the largest integer that is upper bounded by $x$; i.e., $[x]\le x<[x]+1$) 
\begin{align}
\varliminf_{n\to\infty}\frac{1}{n}C_{(1)}^\epsilon\left(\mathcal{N}^{\otimes n}\right)&\ge\varliminf_{n\to\infty}\frac{1}{\left(\left[\frac{n}{k}\right]+1\right)k}C_{(1)}^\epsilon\left(\mathcal{N}^{\otimes \left(\left[\frac{n}{k}\right]\times k\right)}\right)\nonumber\\
&\ge\varliminf_{n\to\infty}\frac{1}{\left(\left[\frac{n}{k}\right]+1\right)k}D_h^{\omega}\left[\left((\mathcal{N}^{\otimes k}\otimes\mathcal{I}_{N})\left(\sigma_{S_{\rm in}^{\otimes k}N}\right)\right)^{\otimes \left[\frac{n}{k}\right]}\,\middle\|\,\left(\mathcal{N}^{\otimes k}\left(\sigma_{S_{\rm in}^{\otimes k}}\right)\otimes\sigma_{N}\right)^{\otimes \left[\frac{n}{k}\right]}\right]\nonumber\\
&=\frac{1}{k}\varliminf_{l\to\infty}\frac{l}{l+1}\times\frac{1}{l}D_h^{\omega}\left[\left((\mathcal{N}^{\otimes k}\otimes\mathcal{I}_{N})\left(\sigma_{S_{\rm in}^{\otimes k}N}\right)\right)^{\otimes l}\,\middle\|\,\left(\mathcal{N}^{\otimes k}\left(\sigma_{S_{\rm in}^{\otimes k}}\right)\otimes\sigma_{N}\right)^{\otimes l}\right]\nonumber\\
&=\frac{1}{k}S\left[(\mathcal{N}^{\otimes k}\otimes\mathcal{I}_{N})\left(\sigma_{S_{\rm in}^{\otimes k}N}\right)\,\middle\|\,\mathcal{N}^{\otimes k}\left(\sigma_{S_{\rm in}^{\otimes k}}\right)\otimes\sigma_{N}\right].
\end{align}
Recall that for a real-valued sequence $\{a_n\}_{n=1}^\infty$ we have $\varliminf_{n\to\infty}a_n\coloneqq\lim_{n\to\infty}\inf_{m\ge n}a_m = \lim_{n\to\infty}\inf_{m\ge \beta n}a_m$ for every $\beta\in\mathbb{N}$ (i.e., if a sequence converges, which is now the sequence $\{\inf_{m\ge n}a_m\}_{n=1}^\infty$, then an infinite subsequence will also converge to the same limit), meaning that $\varliminf_{n\to\infty}a_{\left[\frac{n}{k}\right]} = \lim_{n\to\infty}\inf_{m\ge kn}a_{\left[\frac{m}{k}\right]} =\lim_{n\to\infty}\inf_{l\ge n}a_{l} = \varliminf_{n\to\infty}a_{n}$ and hence explaining the third line. 
The last line follows from Lemma~\ref{AppLemma:QStein}.
From here we conclude that, for every $0<\theta<\CY{1/2}$,
\begin{align}
C(\mathcal{N})\ge\varlimsup_{k\to\infty}\frac{1}{k}\max_{N,\sigma_{S_{\rm in}^{\otimes k}N}}S\left[(\mathcal{N}^{\otimes k}\otimes\mathcal{I}_{N})\left(\sigma_{S_{\rm in}^{\otimes k}N}\right)\,\middle\|\,\mathcal{N}^{\otimes k}\left(\sigma_{S_{\rm in}^{\otimes k}}\right)\otimes\sigma_{N}\right]=\varlimsup_{k\to\infty}\frac{1}{k}\chi\left(\mathcal{N}^{\otimes k}\right)\ge\varlimsup_{k\to\infty}\frac{1}{k}\overline{\chi}^\theta\left(\mathcal{N}^{\otimes k}\right).
\end{align}
Then, combining with Eq.~\eqref{AppEq:HSWUpperBound}, we have 
\begin{align}
\varlimsup_{k\to\infty}\frac{1}{k}\overline{\chi}^\theta\left(\mathcal{N}^{\otimes k}\right)\le\varlimsup_{k\to\infty}\frac{1}{k}\chi\left(\mathcal{N}^{\otimes k}\right)\le C(\mathcal{N})\le\varliminf_{k\to\infty}\frac{1}{k}\overline{\chi}^\theta\left(\mathcal{N}^{\otimes k}\right)\le \varliminf_{k\to\infty}\frac{1}{k}\chi\left(\mathcal{N}^{\otimes k}\right).
\end{align}
Since the limit inferior is upper bounded by the limit superior, we learn that all these quantities are equal to each other, which further implies the existences of $\lim_{k\to\infty}\frac{1}{k}\overline{\chi}^\theta\left(\mathcal{N}^{\otimes k}\right)$ and $\lim_{k\to\infty}\frac{1}{k}\chi\left(\mathcal{N}^{\otimes k}\right)$.
The proof is completed.
\end{proof}
\end{widetext}

\CY{\subsection{Ability to Generate Informational Non-equilibrium Cannot Enhance Asymptotic Classical Communication}\label{Sec:no-go}}
It is worth mentioning that the above result implies a {\em strong converse property} of $\overline{\chi}^\theta$ (see., e.g., Ref.~\cite{Takagi2020}).
To illustrate what it means, let us consider \CY{the following} figure-of-merits:
\begin{align}
&\overline{\chi}_{\Delta}(\mathcal{N})\coloneqq\sup_{0<\theta<\Delta}\lim_{k\to\infty}\frac{1}{k}\overline{\chi}^\theta\left(\mathcal{N}^{\otimes k}\right);\\
&\overline{\chi}_{\rm min}(\mathcal{N})\coloneqq\lim_{\theta\to0}\lim_{k\to\infty}\frac{1}{k}\overline{\chi}^\theta\left(\mathcal{N}^{\otimes k}\right).
\end{align}
Note that they have very different physical meanings.
$\overline{\chi}_{\Delta}(\mathcal{N})$ is the largest possible regularised Holevo information of $\mathcal{N}$ among all possible encoding and decoding schemes that will not \CY{generate informational non-equilibrium}
more than $2\Delta$.
Namely, it is the highest possible value within a given range of out-of-equilibrium effects.
On the other hand, $\overline{\chi}_{\rm min}(\mathcal{N})$ is the regularised Holevo information of $\mathcal{N}$ when we further request that asymptotically the encoding and decoding must have vanishing ability to \CY{generate informational non-equilibrium.}
Hence, intuitively, $\overline{\chi}_{\rm min}$ has a much stronger physical constraint, and hence, by construction, we have 
$
\overline{\chi}_{\rm min}\le\overline{\chi}_{\Delta}
$
for every $0<\Delta<\CY{1/2}$.
If for some $\Delta$ it is possible to have the converse inequality, namely, $\overline{\chi}_{\Delta}\le\overline{\chi}_{\rm min}$, we say that the {\em strong converse property} of $\overline{\chi}^\theta$ happens for this $\Delta$.
Physically, this means that when we restrict to out-of-equilibrium effect controlled by $\Delta$, the best performance in classical communication can happen when the task {\em cannot} 
\CY{generate informational non-equilibrium}
{\em asymptotically}.
\CY{As a corollary of Proposition~\ref{AppResult:HSWThermo},} we report the following strong converse property:\\

\begin{acorollary}\label{coro:strong converse property chi}
The strong converse property of $\overline{\chi}^\theta$ happens for every $0<\Delta<\CY{1/2}$.
\CY{Namely, for every channel $\mathcal{N}$, 
}
\begin{align}
\CY{\overline{\chi}_{\rm min}(\mathcal{N})=\overline{\chi}_{\Delta}(\mathcal{N})\quad\forall\;0<\Delta<\frac{1}{2}.}
\end{align}
\end{acorollary}
\CY{Physically, this finding provides a no-go result---the ability to generate informational non-equilibrium, no matter how strong it is, {\em cannot} enhance the asymptotic classical capacity.
This further suggests that the ability to {\em preserve} informational non-equilibrium is the resource more relevant to classical communication, which is consistent with the finding reported in Ref.~\cite{Stratton2023}.}

\CYtwo{
Finally, we provide an alternative form of the above no-go result, which is Eq.~(6) stated in the companion paper~\cite{Companion2}.
To start with, for $0\le\theta\le1/2$, define the following figure-of-merit in the one-shot regime:
\begin{eqnarray}
\begin{aligned}
C_{(1)|\text{$\theta$-equi}}^\epsilon(\mathcal{N})\coloneqq\max_{M,\Pi,\mathcal{K},\mathcal{L}}&\log_2M\\
{\rm s.t.}\quad&M\in\mathbb{N};\\
&\mathcal{K}\in{\rm CQ}_{M\to S_{\rm in|\Pi}};\\
&\mathcal{L}\in{\rm QC}_{S_{\rm out|\Pi}\to M};\\
&P_s(\mathcal{L}\circ\mathcal{N}\circ\mathcal{K})\ge1-\epsilon;\\
&\norm{\mathcal{L}\circ\mathcal{N}\circ\mathcal{K}\left(\frac{\id_M}{M}\right) - \frac{\id_M}{M}}_1\le2\theta.
\end{aligned}
\end{eqnarray}
By Fact~\ref{AppLemma:AlternativeCCapacity} [and Eq.~\eqref{Eq:useful formula C Theta}], this is $C_{(1)}^\epsilon$ subject to an additional thermodynamic constraint (i.e., the last line in the above equation).
Physically, it is the amount of information that $\mathcal{N}$ can transmit (in the one-shot regime) when encoding and decoding can generate informational non-equilibrium at most to the order $O(\theta)$.
To this reason, we call $C_{(1)|\text{$\theta$-equi}}^\epsilon$ the {\em one-shot $\theta$-equilibrium capacity}.
Then, we have the following bounds:}
\CYtwo{
\begin{aproposition}\label{prop:theta-equi bounds}
{\em (Bounding the One-Shot $\theta$-Equilibrium Capacity)}
Consider a channel $\mathcal{N}$ and a parameter \mbox{$0<\epsilon<(1-1/\sqrt{2})/2$}.
Then we have
\begin{align}
C_{(1)}^\epsilon(\mathcal{N})\le C_{(1)|\text{\rm $\theta$-equi}}^\epsilon(\mathcal{N})\le \frac{W_{\rm corr|\Theta_{\rm C}, (1)}^{2\epsilon}(\mathcal{N})}{k_BT\ln2}
\end{align}
for every $\epsilon\le\theta<1/2$.
\end{aproposition}
\begin{proof}
Using Eq.~\eqref{Eq:max01} and the fact that $\epsilon\le\theta$, one can directly see $C_{(1)}^\epsilon(\mathcal{N})\le C_{(1)|\text{$\theta$-equi}}^\epsilon(\mathcal{N})$.
Following the same argument that proves Eqs.~\eqref{Eq:Computation05} and~\eqref{Eq:Computation06}, one can show that $C_{(1)|\text{$\theta$-equi}}^\epsilon(\mathcal{N})$ is upper bounded by
\begin{widetext}
\begin{eqnarray}
\begin{aligned}
\max_{M,\mathcal{K},\mathcal{L}}\quad&D_0^{2\epsilon}\left[\left((\mathcal{L}\circ\mathcal{N}\circ\mathcal{K})\otimes\mathcal{I}_{M'}\right)\left(\Phi_{MM'}\right)\,\middle\|\,(\mathcal{L}\circ\mathcal{N}\circ\mathcal{K})\left(\frac{\id_{M}}{M}\right)\otimes\frac{\id_{M'}}{M}\right]\\
{\rm s.t.}\quad&\norm{\mathcal{L}\circ\mathcal{N}\circ\mathcal{K}\left(\frac{\id_M}{M}\right) - \frac{\id_M}{M}}_1\le2\theta;\quad\mathcal{K}\in{\rm CQ}_{M\to S_{\rm in|\mathcal{N}}};\quad\mathcal{L}\in{\rm QC}_{S_{\rm out|\mathcal{N}}\to M};\quad M\in\mathbb{N}.
\end{aligned}
\end{eqnarray}
\end{widetext}
Since $0<2\epsilon<1-1/\sqrt{2}$, we can use Eq.~\eqref{AppEq:App:Proof-Com} to further upper bound it by
\begin{eqnarray}
\begin{aligned}
\max_{M,\mathcal{K},\mathcal{L}}\quad&\frac{W_{\rm corr,(1)}^{2\epsilon}\left[\left((\mathcal{L}\circ\mathcal{N}\circ\mathcal{K})\otimes\mathcal{I}_{M'}\right)\left(\Phi_{MM'}\right)\right]}{k_BT\ln2}\\
{\rm s.t.}\quad&\norm{\mathcal{L}\circ\mathcal{N}\circ\mathcal{K}\left(\frac{\id_M}{M}\right) - \frac{\id_M}{M}}_1\le2\theta;\\
&\mathcal{K}\in{\rm CQ}_{M\to S_{\rm in|\mathcal{N}}};\quad\mathcal{L}\in{\rm QC}_{S_{\rm out|\mathcal{N}}\to M};\quad M\in\mathbb{N}.
\end{aligned}
\end{eqnarray}
This is further upper bounded by $W_{\rm corr|\Theta_{\rm C},(1)}^{2\epsilon}(\mathcal{N})/(k_BT\ln2)$ due to Eq.~\eqref{Eq: one-shot transmitted energy}.
\end{proof}
}

\CYtwo{
We can now discuss the strong converse property with Proposition~\ref{prop:theta-equi bounds}.
To this end, let us define the following figure-of-merits in the asymptotic regime.
First, 
\begin{align}\label{Eq: theta-equi iid}
C_{\text{\rm$\theta$-equi}}(\mathcal{N})\coloneqq\lim_{\epsilon\to0}\varliminf_{k\to\infty}\frac{1}{k}C_{(1)|\text{\rm $\theta$-equi}}^\epsilon\left(\mathcal{N}^{\otimes k}\right)
\end{align}
is the asymptotic form of $C_{(1)|\text{\rm $\theta$-equi}}^\epsilon(\mathcal{N})$.
Furthermore,
\begin{align}
C_{\rm max}(\mathcal{N})&\coloneqq\sup_{0<\theta<1/2}C_{\text{\rm$\theta$-equi}}(\mathcal{N});\\
C_{\rm min}(\mathcal{N})&\coloneqq\lim_{\theta\to0}C_{\text{\rm$\theta$-equi}}(\mathcal{N}).
\end{align}
Again, these two asymptotic quantities have very different physical meanings.
$C_{\rm max}(\mathcal{N})$ is the optimal asymptotic information transmission {\em without} any constraint on the ability to generate informational non-equilibrium. 
On the other hand, $C_{\rm min}(\mathcal{N})$ is the one requiring {\em no ability to asymptotically} generate any informational non-equilibrium. 
Consequently, by definitions, we have $C_{\rm min}(\mathcal{N})\le C_{\rm max}(\mathcal{N})$.
When the opposite inequality also holds, it can be viewed, again, as a strong converse property of $C_{\text{\rm$\theta$-equi}}$.
Now, consider a given $0<\theta<1/2$.
When $\epsilon$ is small enough, using Proposition~\ref{prop:theta-equi bounds} and Theorem~\ref{Result:TCTCI} implies, for {\em every} $k\in\mathbb{N}$, 
\begin{align}
C_{(1)}^\epsilon\left(\mathcal{N}^{\otimes k}\right)\le C_{(1)|\text{\rm $\theta$-equi}}^\epsilon\left(\mathcal{N}^{\otimes k}\right)\le C_{(1)}^{3\epsilon}\left(\mathcal{N}\right) + O(\log_2\epsilon),
\end{align}
where $O(\log_2\epsilon)$ is an one-shot error term independent of the copy number $k$.
Dividing everything by $k$, letting \mbox{$k\to\infty$}, and then 
setting $\epsilon\to0$ give [see also Eqs.~\eqref{Eq: one-shot error meaning},~\eqref{Eq:Asymptotic C def} and~\eqref{Eq: theta-equi iid}]
\begin{align}
C(\mathcal{N}) = C_{\text{\rm$\theta$-equi}}(\mathcal{N})\quad\forall\,0<\theta<1/2.
\end{align}
We thus just prove the following result:
\begin{acorollary}\label{coro:no-go}
For every channel $\mathcal{N}$, we have that
\begin{align}
C_{\rm min}(\mathcal{N})=C_{\rm max}(\mathcal{N})=C(\mathcal{N}).
\end{align}
\end{acorollary}
Namely, the ability to generate informational non-equilibrium is {\em not useful} for transmitting classical information in the asymptotic regime.
}

\section{Conclusion}\label{Sec:Conclusion}
\CY{This work aims to uncover the link between transmitting information and energy.
By utilising entropic quantities [Eqs.~\eqref{AppEq:D0epsilon} and~\eqref{AppEq:HypothesisTesting}] to characterise one-shot information transmission (Theorems~\ref{AppThm:MainResult} and~\ref{AppThm:CC}) and energy transmission [Eqs.~\eqref{AppEq:App:Proof-Com} and~\eqref{Eq: one-shot transmitted energy}], we can connect these two seemingly unrelated aspects of quantum dynamics in the one-shot regime (Theorem~\ref{Result:TCTCI}). 
Interestingly, our results further uncover a dynamical version of Landauer's principle (Corollaries~\ref{coro:weak dynamical Landauer} and~\ref{coro:strong dynamical Landauer}), a thermodynamic meaning of the Holevo-Schumacher-Westmoreland theorem (Proposition~\ref{AppResult:HSWThermo}), and a series of no-go results (Corollaries~\ref{coro:strong converse property chi} and~\ref{coro:no-go}).
See also the companion paper~\cite{Companion2} for further discussions.}

\CY{Several questions remain open, and we list a few here.
First, can we apply a similar approach to a state's carriable amount of classical information under thermodynamic constraints~\cite{Narasimhachar2019,Korzekwa2019,Biswas2021}?
Second, inspired by the recent works Ref.~\cite{Shu2019,Rubino2024,Hsieh2024-3}, can we simplify the energy transmission tasks to become more robust to experimental noise and practical imperfection? 
Third, can we utilise a similar approach to uncover the thermodynamic interpretation of other types of information processing tasks such as stochastic distillation via post-selection~\cite{Hsieh2023,Ku2023,Ku2022,Hsieh2024,Hsieh2024-2,Regula2022Quantum,Regula2022PRL,Takagi2024PRA,Yuan2024PRL,Ji2024}, device-independent tasks~\cite{Chen2024PRL,Hsieh2024,Hsieh2024-2}, and exclusion tasks~\cite{Ducuara2022,Ducuara2023,Ducuara2023-2,Hsieh2023-2,Stratton2024}?
Finally, can we apply our approach to uncover similar links between transmitting information and the quantum system's symmetrical properties~\cite{Cavina2024PRL}, temperature~\cite{Lipka-Bartosik2023PRL}, anomalous energy flow~\cite{Hsieh2023IP,Lipka-Bartosik2024PRL}, (unspeakable) coherence~\cite{Shiraishi2024PRL}, and (different notions of) incompatibility~\cite{Hsieh2022PRR,Hsieh2024Quantum,Hsieh2024,Buscemi2020PRL,Ji2024PRXQ,Haapasalo2021,Mitra2023PRA,Heinosaari2016,Heinosaari2014,Mitra2022PRA}?
We leave these open questions for future research, and
we hope our results can initiate people's interest in the interplay of communication and thermodynamics.}

\appendix

\CYtwo{
\section{Mathematical Properties of $D_0^\epsilon$}\label{App}
As the entropic quantity $D_0^\epsilon$ [defined in Eq.~\eqref{AppEq:D0epsilon}] plays a crucial role in this work,
we
thoroughly discuss its mathematical properties here.}
\CYtwo{
To begin with, consider two states $\rho,\sigma$ whose supports have a non-vanishing intersection. 
Their {\em min-relative entropy}~\cite{Datta2009} is defined as
\begin{align}
D_{\rm min}(\rho\,\|\,\sigma)\coloneqq\log_2\frac{1}{{\rm tr}\left(\sigma\Pi_\rho\right)},
\end{align}
where $\Pi_\rho$ is the projector onto $\rho$'s support.
It is related to $D_0^\epsilon$ as follows:
\begin{afact}\label{App: relation with Dmin}
Consider two commuting states $\eta,\xi$. 
Then 
\begin{align}
D_0^\epsilon(\eta\,\|\,\xi) = D_{\rm min}(\eta\,\|\,\xi) \quad\forall\,0<\epsilon\le\mu_{\rm min}(\eta),
\end{align}
where $\mu_{\rm min}(\eta)$ is $\eta$'s smallest strictly positive eigenvalue.
\end{afact}
\begin{proof}
Since $\eta$ and $\xi$ are commuting, we can write \mbox{$\eta=\sum_jq_j\proj{j}$} and \mbox{$\xi=\sum_jr_j\proj{j}$} for some common eigenbasis $\{\ket{j}\}_j$.
Let $\Lambda_\eta\coloneqq\{j\,|\,q_j>0\}$.
Then, the projector onto $\eta$'s support reads $\Pi_\eta = \sum_{j\in\Lambda_\eta}\proj{j}$.
Now, if $\Lambda$ satisfies \mbox{$\sum_{j\in\Lambda}q_j>1-\epsilon$} for some \mbox{$0<\epsilon\le\mu_{\rm min}(\eta)$}, then 
\mbox{$
\sum_{j\in\Lambda}q_j>1-\mu_{\rm min}(\eta)
$.}
This implies $\sum_{j\in\Lambda}q_j=1$ [if not, $\eta$ would have at least one strictly positive eigenvalue that is strictly smaller than $\mu_{\rm min}(\eta)$, a contradiction]. 
Hence, we have \mbox{$\Lambda_\eta\subseteq\Lambda$.}
Also, $\Lambda_\eta\subseteq\Lambda$ implies \mbox{$\sum_{j\in\Lambda}q_j=1>1-\epsilon$} \mbox{$\forall\,0<\epsilon\le1$}.
Hence, for every $0<\epsilon\le\mu_{\rm min}(\eta)$, 
\begin{align}
\text{$\Lambda$ satisfies \mbox{$\sum_{j\in\Lambda}q_j>1-\epsilon$}\quad if and only if\quad$\Lambda_\eta\subseteq\Lambda$.}
\end{align}
We thus conclude, for every $0<\epsilon\le\mu_{\rm min}(\eta)$, 
\begin{align}
&D_0^\epsilon(\eta\,\|\,\xi)\coloneqq\max_{\Lambda:\sum_{j\in\Lambda}q_j>1-\epsilon}\log_2\frac{1}{\sum_{j\in\Lambda}r_j}=\max_{\Lambda_\eta\subseteq\Lambda}\log_2\frac{1}{\sum_{j\in\Lambda}r_j}\nonumber\\
&\quad=\log_2\frac{1}{\sum_{j\in\Lambda_\eta}r_j}=\log_2\frac{1}{{\rm tr}\left(\xi\Pi_\eta\right)}=D_{\rm min}(\eta\,\|\,\xi),
\end{align}
which completes the proof.
\end{proof}
Consequently, as introduced in Ref.~\cite{Aberg2013}, $D_0^\epsilon$ can be viewed as a generalisation of the min-relative entropy $D_{\rm min}$.
As a direct corollary, this also means that slightly perturbing the error parameter $\epsilon$ will not change the value of $D_0^\epsilon$:
\begin{afact}\label{App: perturbation lemma}
Consider two commuting states $\eta,\xi$. 
Then
\begin{align}
D_0^\epsilon(\eta\,\|\,\xi) = D_0^{\delta}(\eta\,\|\,\xi) \quad\forall\,\epsilon,\delta\in(0,\mu_{\rm min}(\eta)].
\end{align}
\end{afact}
}

\CYtwo{
Another direct consequence of Fact~\ref{App: relation with Dmin} is the data-processing inequality of $D_0^\epsilon$:
\begin{afact}\label{App: data-processing ineq D_0}
Consider two commuting states $\eta,\xi$ and a quantum-to-classical channel $\mathcal{L}$.
Then we have
\begin{align}
D_0^\epsilon[\mathcal{L}(\eta)\,\|\,\mathcal{L}(\xi)]\le D_0^\epsilon(\eta\,\|\,\xi)
\end{align}
for every
$
0<\epsilon\le\min\{\mu_{\rm min}(\eta);\mu_{\rm min}[\mathcal{L}(\eta)]\}, 
$
\end{afact}
\begin{proof}
Using Fact~\ref{App: relation with Dmin}, we have, for any $\epsilon$ in the given range,
\begin{align}
D_0^\epsilon[\mathcal{L}(\eta)\,\|\,\mathcal{L}(\xi)] &= D_{\rm min}[\mathcal{L}(\eta)\,\|\,\mathcal{L}(\xi)]\nonumber\\
&\le D_{\rm min}(\eta\,\|\,\xi) = D_0^\epsilon(\eta\,\|\,\xi),
\end{align}
where we have used the data-processing inequality of the min-relative entropy $D_{\rm min}$~\cite{Datta2013}.
\end{proof}
}

\CYtwo{
Now, we discuss the smoothness and continuity of $D_0^\epsilon$.
First, it is ``smoothed'' in the following sense:

\begin{afact}\label{App:smooth D_0}
Consider three commuting states $\eta,\chi,\xi$ and a given parameter $0<\epsilon<1/2$.
Suppose $\norm{\eta-\chi}_1<\epsilon$.
Then
\begin{align}
D_0^{\epsilon-\norm{\eta-\chi}_1}(\eta\,\|\,\xi)\le D_0^{\epsilon}(\chi\,\|\,\xi)\le D_0^{\epsilon+\norm{\eta-\chi}_1}(\eta\,\|\,\xi).
\end{align}
When $\epsilon<\mu_{\rm min}(\eta)/2$, both equalities can be achieved as
\begin{align}
D_{0}^\epsilon(\eta\,\|\,\xi)=D_0^{\epsilon\pm\norm{\eta-\chi}_1}(\eta\,\|\,\xi)=D_0^{\epsilon}(\chi\,\|\,\xi).
\end{align}
\end{afact}
\begin{proof}
Write $\eta=\sum_jq_j\proj{j}$, $\xi=\sum_jr_j\proj{j}$, \mbox{$\chi=\sum_js_j\proj{j}$}, and $\Delta\coloneqq\norm{\eta-\chi}_1<\epsilon$.
For every $\Lambda$, we have
$
\left|\sum_{j\in\Lambda}s_j - \sum_{j\in\Lambda}q_{j}\right|\le\sum_{j\in\Lambda}|s_j-q_j|=\norm{\eta-\chi}_1$. 
Hence, $\sum_{j\in\Lambda}q_j-\Delta\le\sum_{j\in\Lambda}s_j\le\sum_{j\in\Lambda}q_j+\Delta$, meaning that
\begin{align}
&D_0^{\epsilon-\Delta}(\eta\,\|\,\xi)=\max_{\Lambda:\sum_{j\in\Lambda}q_j-\Delta>1-\epsilon}\log_2\frac{1}{\sum_{j\in\Lambda}r_j}\nonumber\\
&\quad\le\max_{\Lambda:\sum_{j\in\Lambda}s_j>1-\epsilon}\log_2\frac{1}{\sum_{j\in\Lambda}r_j}=D_0^{\epsilon}(\chi\,\|\,\xi)\nonumber\\
&\quad\le\max_{\Lambda:\sum_{j\in\Lambda}q_j+\Delta>1-\epsilon}\log_2\frac{1}{\sum_{j\in\Lambda}r_j}=D_0^{\epsilon+\Delta}(\eta\,\|\,\xi).
\end{align}
Since $\epsilon<\mu_{\rm min}(\eta)/2$ implies both $\epsilon+\Delta$ and $\epsilon-\Delta$ are in the interval $(0,\mu_{\rm min}(\eta)]$, using Fact~\ref{App: perturbation lemma} completes the proof.
\end{proof}
Consequently, for a fixed pair $\eta,\xi$, when a state $\chi$ is approximating $\eta$ (in the order of $\epsilon$), $D_0^\epsilon(\eta\,\|\,\xi)$ and $D_0^\epsilon(\chi\,\|\,\xi)$ will become the {\em same} when $\epsilon$ is smaller than $\mu_{\rm min}(\eta)/2$.
This also implies that $D_0^\epsilon$ is a ``smoothed'' version of $D_{\rm min}$.
As an example, suppose $\eta = \proj{0}$, $\chi = (1-\delta)\proj{0} + \delta\id/d$ with $0<\delta<1/4$, and $\xi = \id/d$. Then $D_{\rm min}(\eta\,\|\,\xi) = \log_2d$, while $D_{\rm min}(\chi\,\|\,\xi) = 0$ no matter how small $\delta$ is.
On the other hand, $\norm{\eta-\chi}_1<2\delta$.
Setting $\epsilon = 2\delta<1/2=\mu_{\rm min}(\eta)/2$, Fact~\ref{App:smooth D_0} implies that 
$D_0^\epsilon(\chi\,\|\,\xi) = D_0^\epsilon(\eta\,\|\,\xi) = \log_2d$.
This example demonstrates how $D_0^\epsilon$ ``smoothes'' $D_{\rm min}$.}

\CYtwo{
Finally, $D_0^\epsilon$ is continuous in the following sense---in fact, it is {\em Lipschitz continuous}~\cite{Analysis-book} in its second argument:

\begin{afact}\label{App:conti D_0}
Consider three commuting states $\eta,\xi,\zeta$ and a given parameter $0<\epsilon<1$.
Then we have
\begin{align}
\left|2^{-D_0^\epsilon(\eta\,\|\,\xi)} - 2^{-D_0^\epsilon(\eta\,\|\,\zeta)}\right|\le\norm{\xi-\zeta}_1.
\end{align}
\end{afact}
\begin{proof}
Write $\xi = \sum_jr_j\proj{j}$ and $\zeta = \sum_jt_j\proj{j}$.
Let us assume $2^{-D_0^\epsilon(\eta\,\|\,\zeta)}\ge2^{-D_0^\epsilon(\eta\,\|\,\xi)}$ without loss of generality.
Also, suppose $\Lambda_*$ achieves the optimality of $2^{-D_0^\epsilon(\eta\,\|\,\xi)}$; namely, $2^{-D_0^\epsilon(\eta\,\|\,\xi)} = \sum_{j\in\Lambda_*}r_j$.
Then we have
\begin{align}
&\left|2^{-D_0^\epsilon(\eta\,\|\,\xi)} - 2^{-D_0^\epsilon(\eta\,\|\,\zeta)}\right|=2^{-D_0^\epsilon(\eta\,\|\,\zeta)}-2^{-D_0^\epsilon(\eta\,\|\,\xi)}\nonumber\\
&=\min_{\Lambda:\sum_{j\in\Lambda}q_j>1-\epsilon}\sum_{j\in\Lambda}t_j - \sum_{j\in\Lambda_*}r_j\le \sum_{j\in\Lambda_*}(t_j-r_j)\le\norm{\zeta-\xi}_1,
\end{align}
which completes the proof.
\end{proof}}

\CYtwo{
Combining Facts~\ref{App:smooth D_0} and~\ref{App:conti D_0}, we can address the continuity when both arguments are slightly perturbed:

\begin{afact}\label{App:coro}
Consider four commuting states $\eta,\eta',\xi,\xi'$ and a given parameter $0<\epsilon<1/2$.
Suppose $\norm{\eta-\eta'}_1<\epsilon$, $\norm{\xi-\xi'}_1<\epsilon$, and $\epsilon<\mu_{\rm min}(\eta)/2$.
Then we have
\begin{align}
\left|2^{-D_0^\epsilon(\eta\,\|\,\xi)} - 2^{-D_0^\epsilon(\eta'\,\|\,\xi')}\right|<\epsilon.
\end{align}
\end{afact}
\begin{proof}
Since $\epsilon<\mu_{\rm min}(\eta)/2$, Fact~\ref{App:smooth D_0} implies $D_0^\epsilon(\eta'\,\|\,\xi) = D_0^\epsilon(\eta\,\|\,\xi).$
The result follows by applying Fact~\ref{App:conti D_0}.
\end{proof}
}

\section*{Acknowledgements}
We thank (in alphabetical order) Antonio Ac\'in, Alvaro Alhambra, Stefan B$\ddot{\rm a}$uml, Philippe Faist, Yeong-Cherng Liang, Matteo Lostaglio, Jef Pauwels, Mart\'i Perarnau-Llobet, Bartosz Regula, Valerio Scarani, Gabriel Senno, Yaw-Shih Shieh, Paul Skrzypczyk, Jacopo Surace, Gelo Noel M. Tabia, Ryuji Takagi, Philip Taranto, and Armin Tavakoli for fruitful discussions.
We also thank the Quantum Thermodynamics Summer School (23-27 August 2021, Les Diablerets, Switzerland), organised by L\'idia del Rio and Nuriya Nurgalieva, for the inspirational environment that helped me to improve the early version of this work significantly.
\CY{We acknowledge support from} ICFOstepstone (the Marie Sk\l odowska-Curie Co-fund GA665884), the Spanish MINECO (Severo Ochoa SEV-2015-0522), the Government of Spain (FIS2020-TRANQI and Severo Ochoa CEX2019-000910-S), Fundaci\'o Cellex, Fundaci\'o Mir-Puig, Generalitat de Catalunya (SGR1381 and CERCA Programme), the ERC \CY{Advanced Grant (on grants CERQUTE and FLQuant)}, the AXA Chair in Quantum Information Science, the Royal Society through Enhanced Research Expenses (on grant NFQI), \CY{and the Leverhulme Trust Early Career Fellowship (on grant ``Quantum complementarity: a novel resource for quantum science and technologies'' with number ECF-2024-310)}.

\bibliography{Ref.bib}

\end{document}